\documentclass[twoside]{article}

\usepackage[accepted]{aistats2026}
\usepackage[round]{natbib}

\usepackage{algorithm}
\usepackage{algorithmic}
\usepackage{amsmath}
\usepackage{amssymb}
\usepackage{amsthm}
\usepackage{mathtools}
\usepackage{wrapfig}

\newcommand{\RR}{\mathbb{R}}

\newtheorem{thm}{Theorem}[section]
\newtheorem{prob}[thm]{Problem}

\newtheorem{prop}[thm]{Proposition}

\newtheorem{lemm}[thm]{Lemma}

\newtheorem{defn}{Definition}

\usepackage{tcolorbox}
\usepackage{enumerate,enumitem}

\usepackage{graphicx}
\usepackage{caption}
\usepackage{subcaption}
\usepackage{url} 
\usepackage{xurl}       
\usepackage{hyperref}

\usepackage{booktabs}
\usepackage{tabularx}
\usepackage{longtable}
\usepackage{multirow}
\usepackage{array}

\usepackage{rotating}

\definecolor{lightblue}{RGB}{225,236 ,235}
\definecolor{lightgreen}{RGB}{237, 238, 207}
\definecolor{lightyellow}{RGB}{250, 238, 198}
\definecolor{lightpink}{RGB}{255 ,202, 187}
\definecolor{lightcyan}{RGB}{224, 255, 255}

\definecolor{RedLight}{RGB}{209,77,65}
\definecolor{NavyBlue}{RGB}{31,119,180}
\usepackage{wasysym}  

\begin{document}

%

%

\runningtitle{Randomized HyperSteiner}

%
\runningauthor{A. A. Medbouhi, A. García-Castellanos, G. L. Marchetti, D. M. Pelt, E. J. Bekkers, D. Kragic}

\twocolumn[

\aistatstitle{Randomized HyperSteiner: A Stochastic Delaunay Triangulation Heuristic for the Hyperbolic Steiner Minimal Tree}

\aistatsauthor{
   Aniss Aiman Medbouhi$^{*\, 1}$ \qquad  Alejandro García-Castellanos$^{*\, 2}$ \qquad Giovanni Luca Marchetti$^{3}$ \\ \textbf{Daniël M. Pelt}$^{4}$ \qquad \textbf{Erik J. Bekkers}$^{2}$ \qquad \textbf{Danica Kragic}$^{1}$
}

\aistatsaddress{
$^1$School of Electrical Engineering and Computer Science, KTH Royal Institute of Technology\\
$^2$Amsterdam Machine Learning Lab (AMLab), University of Amsterdam \\
  $^3$School of Engineering Sciences, KTH Royal Institute of Technology\\
  $^4$Leiden Institute of Advanced Computer Science, Universiteit Leiden
} 
]

\begin{abstract}

We study the problem of constructing Steiner Minimal Trees (SMTs) in hyperbolic space. Exact SMT computation is NP-hard, and existing hyperbolic heuristics such as HyperSteiner are deterministic and often get trapped in locally suboptimal configurations. We introduce \textbf{Randomized HyperSteiner} (RHS), a stochastic Delaunay triangulation heuristic that incorporates randomness into the expansion process and refines candidate trees via Riemannian gradient descent optimization. Experiments on synthetic data sets and a real-world single-cell transcriptomic data show that RHS outperforms Minimum Spanning Tree (MST), Neighbour Joining, and vanilla HyperSteiner (HS). In near-boundary configurations, RHS can achieve a 32\% reduction in total length over HS, demonstrating its effectiveness and robustness in diverse data regimes.

\end{abstract}
\section{Introduction}
\label{sec:intro}

\begin{figure*}[t!]
\centering
  \begin{subfigure}[t]{0.3\linewidth}
     \centering
     \includegraphics[width=\linewidth]{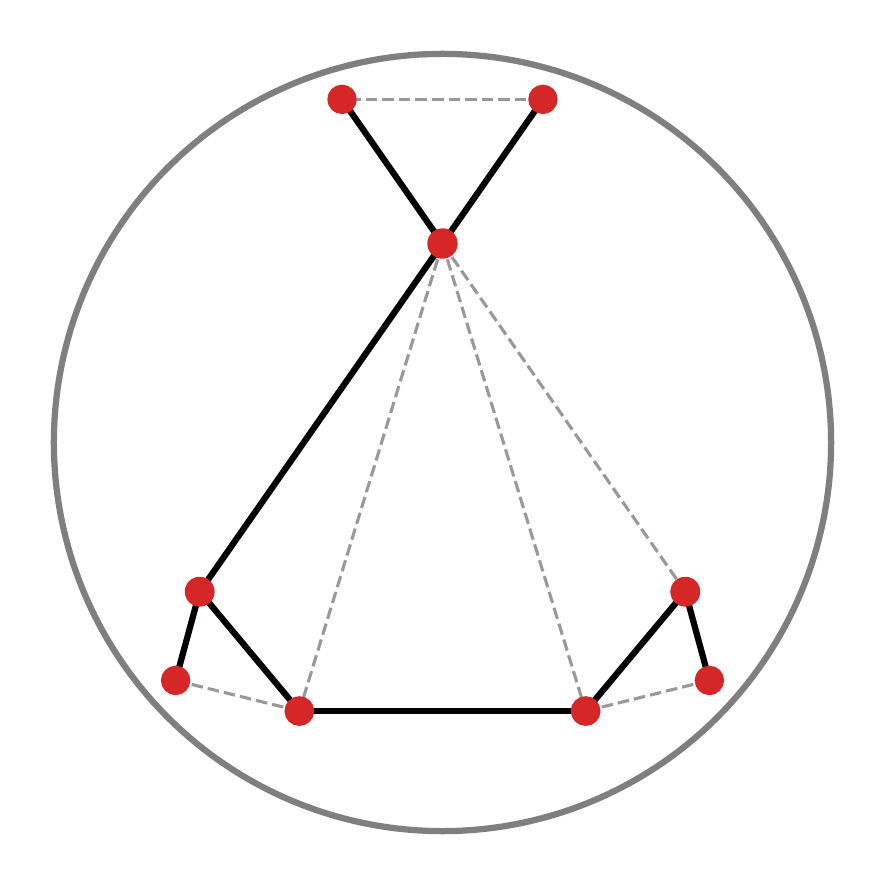}
     \caption*{Minimum Spanning Tree}
     \label{fig:samplesTable1}
 \end{subfigure}
 \hfill
 \begin{subfigure}[t]{0.3\linewidth}
     \centering
     \includegraphics[width=\linewidth]{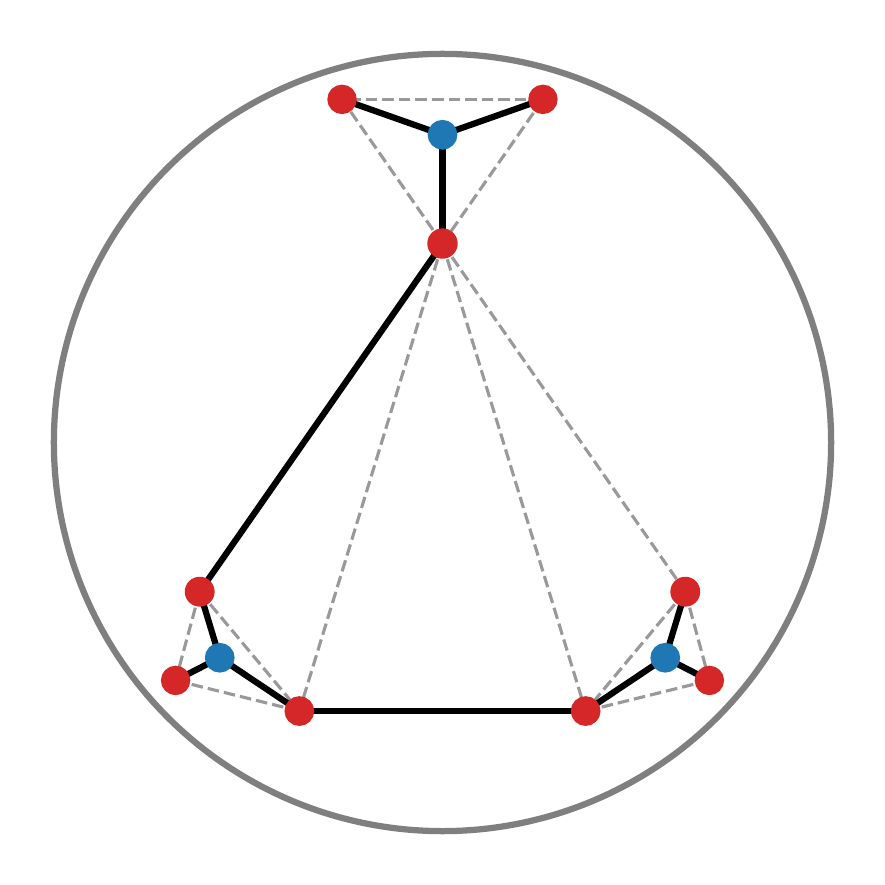}
      \caption*{Vanilla HyperSteiner}
     \label{fig:samplesTable2}
 \end{subfigure}
\hfill
 \begin{subfigure}[t]{0.3\linewidth}        
     \centering
     \includegraphics[width=\linewidth]{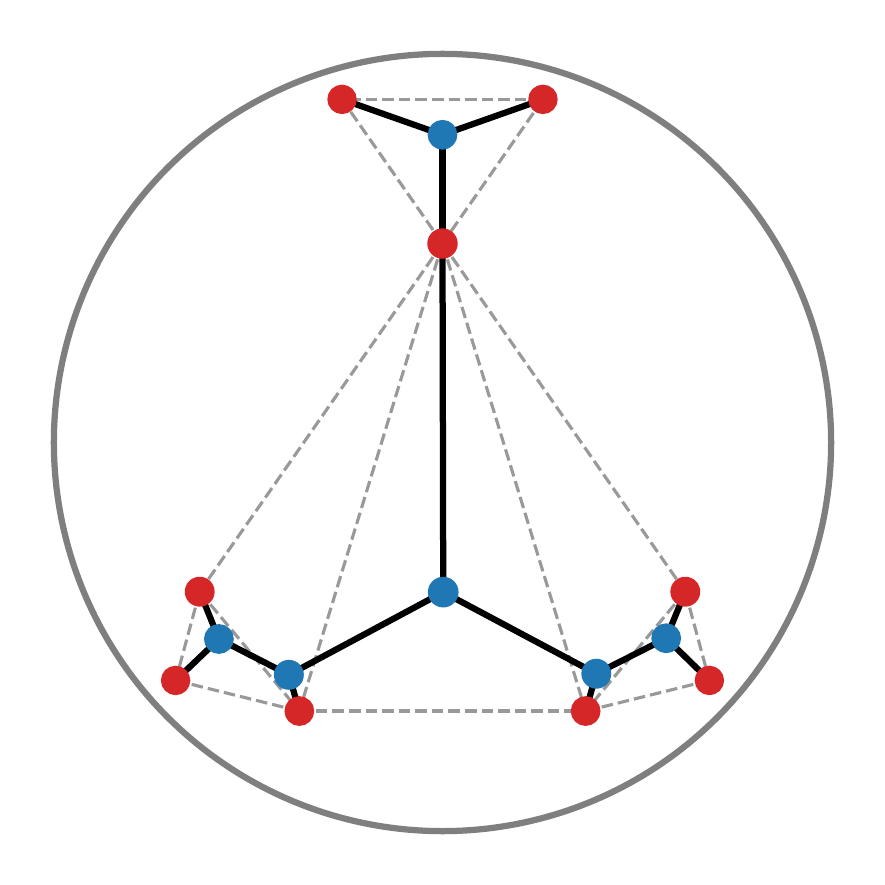}
     \caption*{Randomized HyperSteiner (ours)}
     \label{fig:samplesTable3}
 \end{subfigure}
\caption{Comparison of MST and heuristic SMTs using deterministic versus stochastic sampling over Delaunay triangles. Red ({\color{RedLight}$\newmoon$}) denotes terminals, blue ({\color{NavyBlue}$\newmoon$}) denotes Steiner points, and dashed lines correspond to the auxiliary hyperbolic DT. The deterministic (vanilla) HyperSteiner exhibits myopic behavior, whereas our randomized variant achieves superior global topology and reduces total tree length.}
\label{fig:samplesTables}
\end{figure*}

The principle of \emph{maximum parsimony}, in line with Occam’s razor, states that among competing explanations consistent with the evidence, the simplest should be preferred. In phylogenetics, for example, this principle manifests through tree-based models, where minimality of length provides a natural notion of parsimony \citep{Fitch1971maximumparsimony}. Steiner Minimal Trees (SMTs) formalize this objective by seeking the shortest tree connecting a set of points, potentially augmented with additional \emph{Steiner points}. In contrast to the Minimum Spanning Tree (MST), which only uses the input terminals, SMTs may reduce total length substantially. This makes SMTs a particularly compelling model when inferring hierarchical structures from data, where the ground truth tree is typically unknown, while tree length is an interpretable and measurable quantity that allows systematic comparison across methods. From a broader perspective, our main contribution lies in non-Euclidean computational geometry, with biological hierarchical inference serving only as one potential application of the proposed framework.

\renewcommand{\thefootnote}{\fnsymbol{footnote}}
\footnotetext[1]{\raggedright Equal contribution. E-mails: $<$medbouhi@kth.se$>$, $<$a.garciacastellanos@uva.nl$>$.}

\emph{Hyperbolic geometry} offers a suitable ambient space for tree inference, as it can embed hierarchical and exponentially growing structures with arbitrary low distortion compared to the Euclidean space \citep{sarkar2011low}. This has motivated the development of machine learning models grounded in hyperbolic geometry, with applications in domains such as language \citep{tifrea2018poincar, nickel2017poincare}, biology \citep{Klimovskaia2020SingleCellPoincareMap, Zhou2021HyperbolicGeometryGeneExpression}, games \citep{cetin2022hyperbolic}, robotics \citep{Jaquier2024HGPLVM}, and neurosciences \citep{Zhou2018HyperbolicGeometryOlfactorySpace, Taleb2025TowardsDiscoveringHierarchyOlfactoryHyperbolic}. However, while hyperbolic embeddings effectively capture hierarchy, they do not provide direct algorithms for reconstructing explicit shortest trees embedded within hyperbolic space. Moreover, the SMT problem is NP-hard \citep{garey1977complexity}, making exact solutions infeasible in hyperbolic settings and necessitating heuristic approaches. A first step in this direction was the recently proposed HyperSteiner algorithm \citep{garcia2025hypersteiner}, which adapts a classical Delaunay-based heuristic to hyperbolic geometry, achieving an efficient construction of approximate SMTs. Yet, like other deterministic methods, the vanilla HyperSteiner inherits a major limitation: its MST-based heuristic tends to become trapped in locally optimal but globally suboptimal tree topologies.

In Euclidean space, this limitation has motivated stochastic variants of Delaunay triangulation (DT) heuristics, \citep{Beasley1994DTHeuristicSMT, Laarhoven2011Randomized}, which showed that randomized Delaunay triangles selection can substantially improve solution quality by exploring a broader range of candidate topologies. Extending these ideas to hyperbolic space and combining them with novel theoretical guarantees enabling the use of Riemannian gradient descent \citep{Bonnabel2013RSGD} for global optimization, we revisit HyperSteiner from a stochastic perspective and propose the first randomized SMT heuristic tailored to hyperbolic geometry. As illustrated in Figure~\ref{fig:samplesTables}, while the vanilla HyperSteiner (center) improves upon the MST baseline (left) through local Steiner point insertion, it becomes \emph{myopic} due to its deterministic MST-based strategy, preventing exploration of alternative tree structures that could yield significantly shorter total lengths (right). Our work addresses this challenge by introducing stochasticity, thereby broadening topological exploration and reducing the risk of suboptimal solutions, yielding shorter tree configurations. Our contributions are as follows:

\vspace{-7pt}
\begin{itemize}
    \item We propose \textbf{Randomized HyperSteiner} (RHS), a stochastic Delaunay triangulation heuristic for constructing Steiner Minimal Trees in hyperbolic space.
    \item We extend the Euclidean uniqueness theorem to hyperbolic space, proving that Steiner trees have a unique minimizing configuration for fixed topology and terminals.
    \item We empirically demonstrate that RHS consistently outperforms Minimum Spanning Tree, Neighbour Joining, and original HyperSteiner across synthetic and real datasets.
    \item We show that RHS achieves reductions over MST values up to 43\%, approaching the theoretical upper bound of 50\% in near-boundary settings.
\end{itemize}

\section{Related work}
\label{sec:rel}

We review Delaunay triangulation heuristics for Euclidean Steiner Minimal Trees, non-Euclidean generalizations, and hierarchical inference methods. We add a review on hyperbolic machine learning in the Appendix~(Section~\ref{app:hyp_ml}).

\textbf{Euclidean Delaunay-Based Heuristic SMTs.} Heuristic methods for Euclidean Steiner Minimal Trees (SMTs) commonly exploit geometric structures, particularly Delaunay triangulations (DT), to identify promising local configurations for optimization \citep{zachariasen_concatenation-based_1999}. These algorithms typically follow a three-phase structure: \emph{Expansion} selects triangles from the DT as candidates for local 3-point Steiner problems (i.e., Fermat points); \emph{Heuristic Construction} assembles local solutions into a global tree; and \emph{Exploration Policy} determines refinement strategies.
The seminal Smith-Lee-Liebman (SLL) algorithm \citep{smith1981n} exemplifies this framework through deterministic MST-based expansion. SLL selects Delaunay triangles containing two MST edges as candidates for Steiner point insertion. When all internal angles are below 120°, the triangle is replaced by its Fermat point, minimizing total edge length while ensuring 120° edge angles \citep{Gilbert1968SMT}. This yields locally-optimal full Steiner topologies with high computational efficiency.
However, deterministic expansion strategies suffer from limited topological exploration, often becoming trapped in local minima for clustered or symmetric point configurations \citep{Laarhoven2011Randomized}. To address this limitation, researchers have introduced stochasticity at multiple algorithmic levels. \citet{Beasley1994DTHeuristicSMT} pioneered this direction with two key innovations: iterative expansion that recomputes the DT after each Steiner point addition, enabling dynamic geometric adaptation; and simulated annealing as stochastic exploration policy. Beasley's method examines all DT triangles at each expansion step, ensuring comprehensive local configuration analysis. \citet{Laarhoven2011Randomized} further advanced this approach by demonstrating that stochastic triangle sampling achieves comparable results with significantly reduced computational cost. Instead of exhaustively processing all triangles, their method randomly samples Delaunay simplices during expansion, improving solution quality through broader topological exploration while maintaining efficiency. Similarly, our proposed approach employs stochastic exploration to overcome the limitations of deterministic hyperbolic methods such as \cite{garcia2025hypersteiner}

\textbf{Non-Euclidean SMTs.} Steiner Minimal Tree problems have been extended to various non-Euclidean metric spaces, including rectilinear SMTs in Manhattan metrics \citep{hanan1966steiner}, orientation-dependent metrics \citep{yan1997steiner}, and Riemannian manifolds \citep{ivanov2000steiner, logan2015steiner}. Euclidean heuristics have been successfully adapted to spaces of constant curvature, with SLL algorithm extensions to spherical \citep{dolan1991minimal} and hyperbolic \citep{garcia2025hypersteiner} geometries. Alternative approaches include analytic constructions for hyperbolic Steiner trees using hyperbolic trigonometry in the upper half-plane model \citep{halverson2005steiner}, though these become computationally intractable for larger terminal sets.
To address scalability limitations, \citet{garcia2025hypersteiner} introduced HyperSteiner, an efficient heuristic that adapts Delaunay-based SMT construction to hyperbolic space. Operating in the Klein model, their method employs angle-based insertion criteria for Steiner points, achieving both effectiveness and scalability (see Algorithm~\ref{alg:og_hyper}). However, like its Euclidean counterparts, this deterministic approach suffers from myopic optimization, limiting topological exploration and potentially yielding suboptimal configurations.
Extending \citet{Beasley1994DTHeuristicSMT, Laarhoven2011Randomized}, we introduce a stochastic framework that samples candidate Steiner topologies through hyperbolic Delaunay triangulations and refines them via Riemannian gradient descent \citep{Bonnabel2013RSGD}. As demonstrated in Figure~\ref{fig:samplesTables} (right), our Randomized HyperSteiner overcomes the topological limitations of deterministic methods, discovering superior global configurations through strategic randomization combined with geometric optimization in hyperbolic space.

\textbf{Hierarchical Inference.} Inferring hierarchical structure from observational data is central to fields such as evolutionary biology, phylogenetics, and cell lineage reconstruction, where tree-like relationships describe divergence processes \citep{Cieslik2006shortestconnectivity, Gong20218benchmarkcelllineage}. Traditional methods reconstruct trees from sequence data through pairwise distance estimation---as in Unweighted Paired Group Mean Arithmetic \citep{Sokal1958UPGMA} and Neighbor Joining \citep{Saitou1987neighborjoining}---or probabilistic character-based models such as Maximum Parsimony \citep{Fitch1971maximumparsimony}, Maximum Likelihood \citep{Felsenstein1981MaximumLikelihoodPhylogenetic}, and Bayesian frameworks \citep{Huelsenbeck2001BayesianPhylogenetic}.
When formulated to minimize mutational events, phylogenetic inference can be modeled as a Steiner tree problem \citep{Foulds1979steinerminimalphylogenetictree}. This connection has been explored in Euclidean geometry \citep{Blelloch1997PhylogeneticSteinerTree, Brazil2009PhylogeneticSteinerTree, FA2022steinertreelineagecell} and is now investigable in hyperbolic space through \citet{garcia2025hypersteiner}. Recent work has applied variational Bayesian inference in hyperbolic geometry \citep{Mimori2023Geophy, Macaulay2024Dodonaphy, Chen2025hyperbolicVCSMC}, learning latent representations of molecular data conditioned on biological evolutionary models. However, the resulting phylogenetic trees are not necessarily embedded within hyperbolic space, nor is their construction guided by tree minimality objectives.
In contrast, our method applies generally to arbitrary datasets with hyperbolic representations, constructing minimal trees directly embedded within the same hyperbolic space. This enables hierarchical structure inference compatible with existing hyperbolic embedding techniques, such as those developed for biological data \citep{Klimovskaia2020SingleCellPoincareMap, Zhou2021HyperbolicGeometryGeneExpression}.
\section{Background}
\label{sec:bg}

We focus on the Klein-Beltrami model, Steiner minimal trees, and Delaunay triangulations.

\subsection{The Klein-Beltrami Model}

The hyperbolic space of dimension $n$ is a complete, simply-connected Riemannian manifold with constant negative curvature. Several models exist to represent this geometry. Among them, the Klein-Beltrami model is particularly convenient due to its straight-line geodesics facilitating computations. In this model, points lie in the open unit ball $\mathbb{K}^{n}= \left\{z \in \RR^n \mid\|z\|_2 <1\right\}$, and the Riemannian metric is defined by 
\begin{equation}
 g(z) = \frac{-1}{\langle z,z \rangle}I_{n} + \frac{1}{\langle z,z \rangle^2} z\otimes z,
 \end{equation}
where $\langle \cdot, \cdot \rangle:=-1+\langle \cdot, \cdot \rangle_E$ denotes the Lorentzian Inner Product in homogeneous coordinates (with $\langle \cdot, \cdot \rangle_E$ being the standard Euclidean inner product), $I_n$ the identity matrix, and $\otimes$ the tensor product. The corresponding geodesic distance between two points $x, y \in \mathbb{K}^n$ is given by:
\begin{equation}\label{eq:KleinDistance}
d(x,y)= \operatorname{arccosh}\left(\frac{-\langle x,y\rangle}{\sqrt{\langle x,x \rangle \langle y,y \rangle}}\right).
\end{equation}

\subsection{Steiner Minimal Trees}
\label{sec:steiner_info}

Let $\mathcal{X}$ be a metric space equipped with a distance function $d \colon \mathcal{X} \times \mathcal{X} \to \mathbb{R}_{\geq 0}$. The \emph{Steiner minimal tree} (SMT) problem seeks to find the shortest network connecting a given set of points (the \emph{terminals}), potentially including additional auxiliary points (the \emph{Steiner points}) not in the original input. This can be formally stated as follows:
\begin{prob}\label{prob:mst}
Given a finite set of points $P \subseteq \mathcal{X}$, find a finite undirected connected graph $\text{\rm SMT}(P)=(V, E)$ embedded in $\mathcal{X}$ such that $P \subseteq V$ and the total edge length $L(\text{\rm SMT}(P))$ is minimized, with:
\begin{equation}
L(\text{\rm SMT}(P)) = \sum_{(x,y) \in E} d(x,y).
\end{equation}
\end{prob}
The optimal solution is always acyclic and hence forms a tree. Thus, the SMT problem can be reformulated as finding the tree with minimal length that connects a set of points. The number of Steiner points is bounded by $|P| - 2$, and the tree is said to be a \emph{full Steiner tree} (FST) when this bound is reached. For instance, an FST for three terminals contains a single Steiner point located at the \emph{Fermat point}, while configurations with four terminals admit two Steiner points.

The SMT problem is known to be NP-hard \citep{garey1977complexity}. If no Steiner points are allowed (that is, $V = P$), the solution is simply the \emph{Minimum Spanning Tree} (MST), which can be efficiently computed via algorithms such as the one of \cite{kruskal1956shortest} with computational time complexity of $\mathcal{O}(|P|^2 \log |P|)$.

SMTs always achieve a length less than or equal to that of MSTs. The quality of the improvement is quantified by the \emph{Steiner ratio}.

\begin{defn}[Steiner ratio]\label{locstdef}
For a finite subset $P \subseteq \mathcal{X}$, the \emph{local Steiner ratio} is defined as
\[
\rho(P) = \frac{L(\mathrm{SMT}(P))}{L(\mathrm{MST}(P))} \leq 1.
\]
The \emph{global Steiner ratio} of $\mathcal{X}$ is the infimum of $\rho(P)$ over all finite subsets $P \subseteq \mathcal{X}$.
\end{defn}

The Gilbert–Pollak Conjecture states that for the Euclidean plane $\mathcal{X} = \mathbb{R}^2$, the global Steiner ratio equals $\sqrt{3} / 2$. The conjecture is still open to this date \citep{ivanov2012steiner}. For $n$-dimensional Riemannian manifolds, the global Steiner ratio does not exceed the one of the Euclidean space $\mathbb{R}^n$ of the same dimension \citep{cieslik2001steiner}. Notably, in the hyperbolic plane, the global Steiner ratio has been shown to be equal to $\frac{1}{2}$ \citep{innami_steiner_2006}. In other words, a Steiner tree in hyperbolic space can achieve in theory a total length up to 50\% shorter than that of the corresponding MST, whereas in the Euclidean plane the maximum achievable reduction is approximately 13.4\%.

\subsection{Delaunay Triangulation}
As mentioned in Section~\ref{sec:rel}, we focus on the family of heuristics that computes SMT by restricting the search space of potential tree topologies through \emph{Delaunay triangulation} (DT). Due to its relevance in this work, we recall its definition below in any metric space $\mathcal{X}$ equipped with a distance $d$. 
\begin{defn}
The \emph{Voronoi cell} of $p \in P \subseteq \mathcal{X}$ is defined as:
\begin{equation}
\label{eq:v_cell}
V(p) = \{ x\in \mathcal{X} \ | \ \forall q \in P \  d(x,q) \geq d(x,p) \}.
\end{equation}
The \emph{Delaunay triangulation} is defined as the collection of simplices with vertices $\sigma \subseteq P$ such that ${\bigcap_{p \in \sigma } V(p) \not = \emptyset}$.
\end{defn} 

This definition extends to hyperbolic space by replacing the distance function $d$ with the hyperbolic metric given in Equation~\ref{eq:KleinDistance}.

\section{Method}
\label{sec:meth}

Algorithm~\ref{alg:rad_hyper} summarizes our method. For detailed descriptions of the sub-protocols, see Appendix (Section~\ref{app:heuristic}). The Randomized HyperSteiner follows the three characteristic stages of DT-based methods (Section~\ref{sec:rel}):

\begin{algorithm}[th!]
\caption{Randomized HyperSteiner}
\label{alg:rad_hyper}
\begin{algorithmic}
\vspace{1mm}
\STATE \textbf{Input}: Terminal point set $P\subset\mathbb{K}^2$, maximum iterations $N^*$, insertion probability range $[l,u]$
\vspace{1mm}
\STATE Initialize $n \leftarrow 1$, $S \leftarrow \emptyset$, $S^* \leftarrow \emptyset$, $T^* \leftarrow \text{MST}(P)$.
\WHILE{$n \leq N^*$}
\STATE $\triangleright\, \textsc{\textbf{Fast Iterative Stochastic Expansion}}$\vspace{-5pt}
     \begin{tcolorbox}[colback=lightblue, colframe=cyan!70!black, boxrule=0.5pt]
     \FOR{$\lfloor 2\sqrt{n} - 1 \rfloor$ iterations}
    \STATE Compute hyperbolic Delaunay triangulation $\mathcal{T} = \text{DT}(P \cup S)$
    \STATE Randomly generate the insertion probability, $p$, from the range $[l,u]$.
    \FOR{each triangle $\tau \in \mathcal{T}$}
            \STATE  With probability $p$, compute the \textbf{barycenter} point $b$ of $\tau$ and  let $S = S \cup \{b\}$. 
    
    \ENDFOR
\ENDFOR
 \end{tcolorbox}
    \STATE $\triangleright\, \textsc{\textbf{Heuristic Construction}}$\vspace{-5pt}
\begin{tcolorbox}[colback=lightyellow, colframe=orange!70!black, boxrule=0.5pt]

\STATE $\triangleright\, \textsc{\textbf{Reduction Step}}$  via Algorithm~\ref{alg:red}
 \STATE $S, T \leftarrow \textsc{Degree-Condition}(P, S)$.
 \vspace{2pt}
\STATE \colorbox{lightpink}{Optimize $S'$ via \textbf{Riemannian GD}.} \vspace{5pt}
    
    \STATE $\triangleright\, \textsc{\textbf{Expansion Step}}$  via Algorithm~\ref{alg:edge-in}
    \STATE $S, T \leftarrow \textsc{Angle-Condition}(T, P, S)$.
    \vspace{2pt}
    \STATE \colorbox{lightpink}{Optimize $S'$ via \textbf{Riemannian GD}.}
\vspace{5pt}
    \IF{$L(T) < L(T^*)$}
    \STATE $\triangleright\, \textsc{\textbf{Refinement Step}}$  via Algorithm~\ref{alg:red}
     \STATE $S, T \leftarrow \textsc{Degree-Condition}(P, S)$.
     \vspace{2pt}
\STATE \colorbox{lightpink}{Optimize $S'$ via \textbf{Riemannian GD}.} \vspace{5pt}
    \ENDIF
    \end{tcolorbox}

    \STATE $\triangleright\, \textsc{\textbf{Exploration Policy}}$\vspace{-5pt}
\begin{tcolorbox}[colback=lightgreen, colframe=green!50!black, boxrule=0.5pt]
    \IF{$ L(T) < L(T^*) $}
        \STATE Set $S^*\leftarrow S$, $T^* \leftarrow T$,  $n\leftarrow 1$.
    \ELSE
       \STATE Set $S\leftarrow S^*$,  $n\leftarrow n+1$.
    \ENDIF
\end{tcolorbox}
\ENDWHILE
\vspace{1mm}
\STATE \textbf{Output}: A heuristic SMT $T^*$ with Steiner points~$S^*$ 
\vspace{1mm}
\end{algorithmic}
\end{algorithm}

\vspace{-7pt}

    \paragraph{Expansion.} At iteration $n$, given the current best Steiner points, we iteratively compute the hyperbolic DT and probabilistically sample triangles. Following \citet{Laarhoven2011Randomized}, we initialize Steiner point candidates with hyperbolic barycenters for fast computation. The step function $\lfloor 2\sqrt{n} - 1 \rfloor$ controls the number of expansions per exploration trial.
\vspace{-7pt}
    \paragraph{Heuristic Construction.} Using the expanded candidate set, construct SMT approximations that respect the following properties:  
    \begin{itemize}
        \item \emph{Angle condition:} All angles in an SMT are at least $120^\circ$, with Steiner point angles exactly $120^\circ$.

        \item \emph{Degree condition:} All SMT vertices have degree at most three, with Steiner points having degree exactly three.

        \item \emph{Topology--geometry duality:} Given optimal Steiner point locations $S$, the Steiner tree is $\text{MST}(P\cup S)$. Conversely, given an optimal topology, Steiner point optimization becomes a convex problem minimizing tree length \citep{Laarhoven2011Randomized}.  
        
    \end{itemize}

    \paragraph{Exploration Policy.} If the new tree has shorter length than the best-so-far, update the solution and reset the iteration counter; otherwise, discard the proposal and increment $n$.

We next describe the three core computational components in detail:  
(i) hyperbolic Delaunay triangulation for expansion,  
(ii) local full Steiner tree (FST) solutions for heuristic construction, and  
(iii) global refinement of Steiner points via Riemannian Gradient Descent (GD).  

\subsection{Hyperbolic Delaunay Triangulation}
\label{sec:hyperdel}

The Klein-Beltrami model enables efficient computation of hyperbolic Delaunay triangulations, since geodesics are straight lines. In this model, hyperbolic Voronoi cells can be derived from \emph{Euclidean power diagrams} \citep{HyperbolicVoronoiDiagramsMadeEasy2010}, which generalize standard Voronoi diagrams.

Let $\mathcal{X}$ be a metric space with a distance function $d$. The power cell associated with point $p \in P$ and weight $r_p \geq 0$ is defined as:
\begin{equation*}
R(p) = \{x \in \mathcal{X} \mid d(x, p)^2 - r_p^2 \leq d(x, q)^2 - r_q^2, \ \forall q \in P\}.
\end{equation*}
When all weights vanish, power cells reduce to Voronoi cells.

The Klein-Beltrami model admits a non-isometric embedding into Euclidean space, allowing hyperbolic Voronoi cells to be computed as restricted Euclidean power cells. The following result formalizes this equivalence:
\begin{thm}[Nielsen and Nock, 2009]
Given $P \subseteq \mathbb{K}^n$, there exists an explicit set $\widetilde{P} \subseteq \RR^n $ and weights $\{ r_{\widetilde{p}}\}_{\widetilde{p} \in \widetilde{P}}$ such that the hyperbolic Voronoi cells of $P$ correspond to restrictions to $\mathbb{K}^n$ of power cells of $\widetilde{P}$.
\end{thm}

Specifically, the set $\widetilde{P}$ (for computing power cells) and weights $\{ r_{\widetilde{p}}\}_{\widetilde{p} \in \widetilde{P}}$ are derived from points $p \in P$ as follows:
\begin{equation}
\widetilde{p} = \frac{p}{2\sqrt{1- \|p\|^2}}, \hspace{3em} r_{\widetilde{p}}^2=\frac{\|p\|^2}{4(1- \|p\|^2)} - \frac{1}{\sqrt{1- \|p\|^2}}. 
\end{equation}

Since power diagrams are efficiently computable in Euclidean space, this allows Delaunay triangulations in hyperbolic geometry to be obtained using standard computational geometry tools.

\subsection{Local FST Solutions via Fermat Points}
\label{sec:fermatpts}

Local refinement of candidate Steiner trees can be reduced to solving small full Steiner tree (FST) subproblems. The previously presented \textit{angle condition}, i.e., incident edges at Steiner points meet at $120^\circ$, motivates the introduction of \emph{isoptic curves}: 

\begin{defn}
\label{def:isoptic}
Let $\mathcal{X}$ be a complete Riemannian surface and $x,y \in \mathcal{X}$. The \emph{isoptic curve} with angle $\alpha$, denoted $C_\alpha(x,y)$, is the locus of points $s \in \mathcal{X}$ such that a geodesic from $s$ to $x$ forms an angle $\alpha$ at $s$ with a geodesic from $s$ to $y$.  
\end{defn}

Isoptic curves reduce to circular arcs when $\mathcal{X} = \mathbb{R}^2$. As a direct consequence, Fermat points can be characterized by intersecting three isoptic curves at angle $2\pi/3$: for terminals $P=\{x,y,z\} \subseteq \mathcal{X}$, the Fermat point $s$ satisfies
\begin{equation}\label{eq:inters}
\{ s\} = C_{\frac{2 \pi}{3}}(x,y) \cap C_{\frac{2 \pi}{3}}(y,z) \cap C_{\frac{2 \pi}{3}}(z,x).
\end{equation}
This generalizes the classical Torricelli construction in the Euclidean plane (Figure~\ref{fig:isopticsEuc}) to hyperbolic geometry (Figure~\ref{fig:isopticsKlein}).  

\begin{figure}[h!]
\centering
  \begin{subfigure}[t]{0.5\linewidth}
     \centering
     \includegraphics[width=\linewidth]{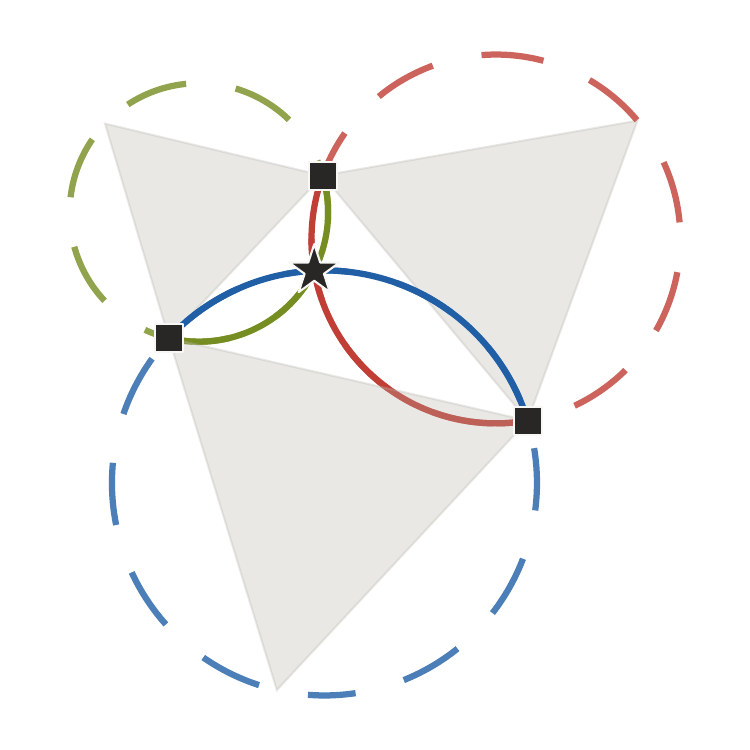}
     \caption{Euclidean, with Torricelli's construction}
     \label{fig:isopticsEuc}
 \end{subfigure}\hfill
 \begin{subfigure}[t]{0.5\linewidth}
     \centering
     \includegraphics[width=\linewidth]{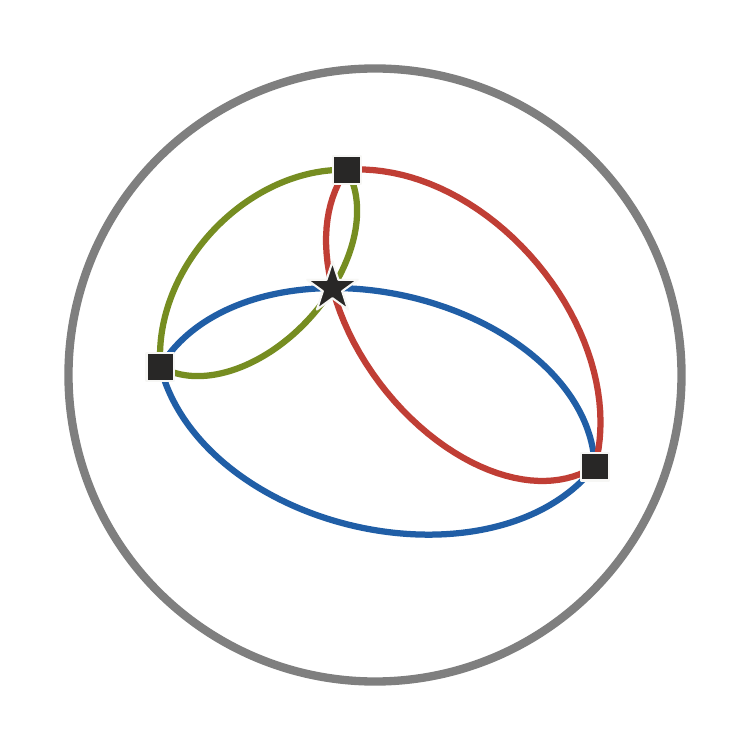}
     \caption{Klein-Beltrami}
     \label{fig:isopticsKlein}
 \end{subfigure}
\caption{Example of isoptic curves for $\alpha=2\pi/3$. Each color represents a different curve. Square ($\,\blacksquare\,$) denotes terminals while star ($\bigstar$) denotes Steiner points. Source: \cite{garcia2025hypersteiner}.}
\label{fig:isoptics}
\end{figure}

We now provide an explicit formula for isoptic curves in the Klein-Beltrami model.  

\begin{prop}[García-Castellanos\,et\,al.,\,2025]\label{prop:isoptic}

Given two points $x, y \in \mathcal{X}=\mathbb{K}^2$ and an angle $0<\alpha<\pi$, the isoptic curve $C_\alpha(x,y)$ in the Klein-Beltrami model is given by $\varphi_{x,y, \alpha}(s) = 0$, where: 
 \begin{align*}
&  \varphi_{x,y, \alpha}(s) =  \langle x, s\rangle \langle y,s \rangle - \langle x,y \rangle \langle s,s \rangle - \\
&- \cos(\alpha)\sqrt{\left(\langle x,s \rangle^2 - \langle x,x \rangle\langle s,s \rangle \right) \left(\langle y,s \rangle^2 - \langle y,y \rangle\langle s,s \rangle \right)}.
\end{align*}
\end{prop}
As the intersection of any two isoptic curves suffices to determine the Fermat point, computing the point $s$ for a terminal set $P = \{x, y, z\} \subseteq \mathbb{K}^2$ reduces to solving a system of two equations, such as:
\begin{equation}\label{eq:system}
\begin{cases}
\varphi_{x,y,\frac{2\pi}{3}}(s) = 0 \\
\varphi_{y,z,\frac{2\pi}{3}}(s) = 0.
\end{cases}
\end{equation}
The resulting nonlinear system can be solved numerically in $\mathcal{O}(1)$ time (see Appendix \ref{app:scala} for details), for example via the hybrid algorithm proposed by \cite{Powell1970NumericalMethodsNonlinearAlgebraicEquations}, which extends Newton's method to higher dimensions for root finding. 

\subsection{Global Refinement via Riemannian GD}
\label{sec:glob_ref}

While local FST corrections improve topology, they may yield suboptimal global configurations. To address this, we refine Steiner positions by minimizing total tree length over the manifold.  

Given a candidate tree $T=(V,E)$ with terminals $P$ and Steiner set $S = V \setminus P$, the objective is
\begin{equation}
\min_{S \subset \mathbb{K}^2} L(S)= \min_{S \subset \mathbb{K}^2}\sum_{(u,v) \in E} d_{\mathbb{K}}(u,v).
\end{equation}
This non-linear optimization problem can be solved via Riemannian GD \citep{Bonnabel2013RSGD}, which performs updates
\begin{equation}
S^{(t+1)} = \exp_{S^{(t)}}\left(-\eta \nabla L\left(S^{(t)}\right)\right),
\end{equation}
where $\exp$ denotes the Riemannian exponential map (applied to each Steiner point) and $\nabla L$ the Riemannian gradient of total tree length. This ensures descent along geodesics, extending Euclidean GD to hyperbolic space. Convergence properties mirror the Euclidean case presented in \cite{Gilbert1968SMT}. Indeed, we show the following uniqueness result:
\begin{thm}
\label{thm:uniqueness}
Let $P_1,\dots,P_n$ be fixed terminal points in the Klein disk model of hyperbolic space.  
Fix a tree topology with $s$ Steiner points $S_1,\dots,S_s$.  
Then there exists a unique placement of the Steiner points $S_1,\dots,S_s$ that minimizes the total tree length $L$.
\end{thm}
Further details and proof of Theorem~\ref{thm:uniqueness} can be found in the Appendix (Section~\ref{app:uniq}).

\section{Experiments}
\label{sec:exp}
\begin{figure*}[t!]
\centering
\begin{subfigure}[t]{0.33\linewidth}
    \centering
    \includegraphics[width=\linewidth]{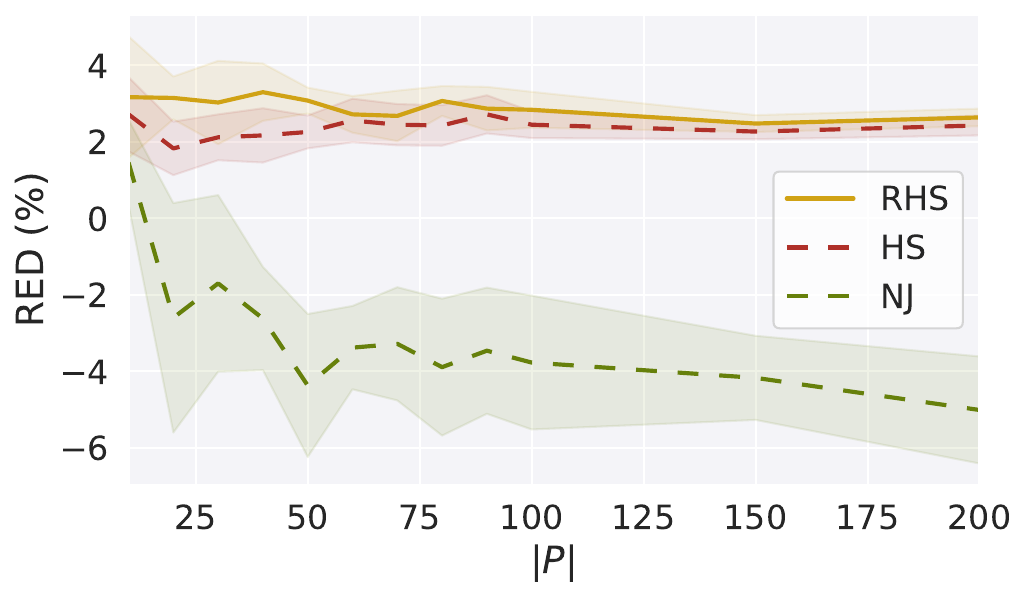}
    \captionsetup{justification=centering}
    \caption{Centered Gaussian $\mathcal{G}(0, 0.5)$}
    \label{fig:row1}
\end{subfigure}%
\hfill
\begin{subfigure}[t]{0.33\linewidth}
    \centering
    \includegraphics[width=\linewidth]{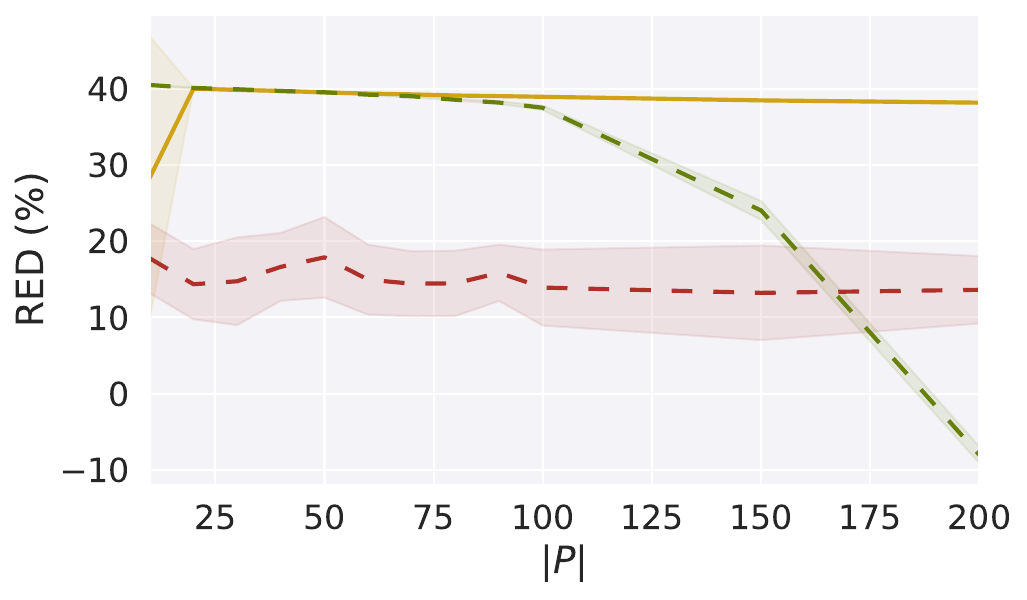}
    \captionsetup{justification=centering}
    \caption{Mixture of $10$ Gaussians $\mathcal{G}(\mu_{10,k}(1-10^{-10}), 0.1)$}
    \label{fig:row2}
\end{subfigure}%
\hfill
\begin{subfigure}[t]{0.33\linewidth}
    \centering
    \includegraphics[width=\linewidth]{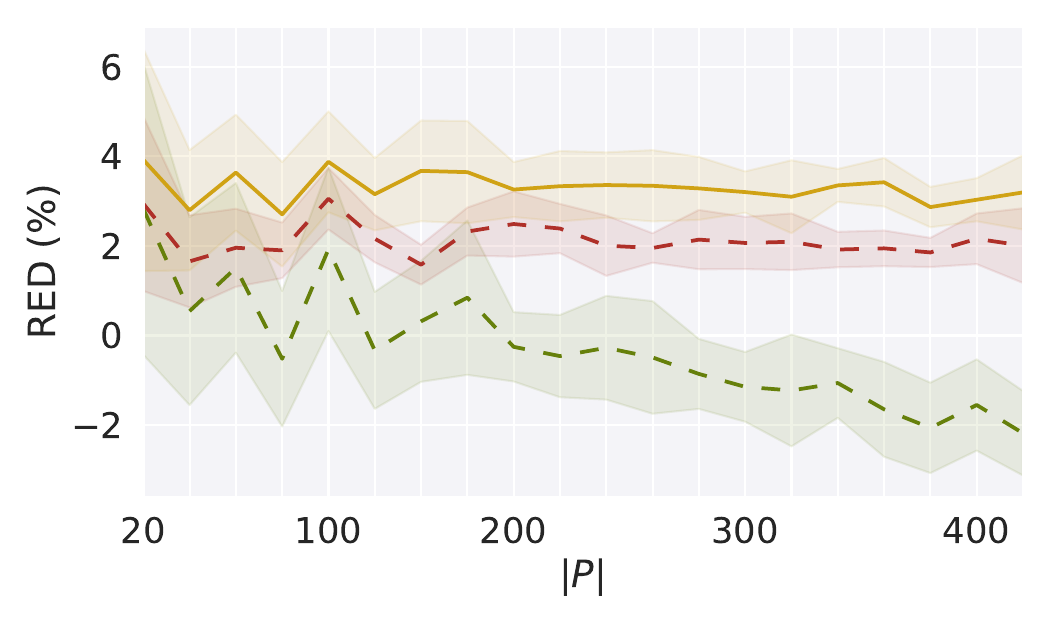}
    \captionsetup{justification=centering}
    \caption{\emph{Planaria} dataset}
    \label{fig:row3}
\end{subfigure}
\caption{
Performance comparison across data distributions. Randomized HyperSteiner (orange), vanilla HyperSteiner (red), and Neighbour Joining (green). Curves report mean tree-length reduction over MST (RED, \%) with standard deviation (shaded), averaged across 10 runs.
}
\label{fig:AcrossDataDistributionsPlot}
\end{figure*}

We empirically evaluate Randomized HyperSteiner (RHS) for constructing hyperbolic Steiner minimal trees with two objectives: (i) quantify tree quality via reduction over the Minimum Spanning Tree (RED, \%), and (ii) benchmark against standard methods—Minimum Spanning Tree (MST), Neighbour Joining (NJ) \citep{Saitou1987neighborjoining}, and the deterministic HyperSteiner (HS) heuristic \citep{garcia2025hypersteiner}. Experiments were conducted on Google Cloud (\textsc{n2d} CPUs, 1 core, 128GB RAM). Source code is available in \url{https://github.com/AGarciaCast/RandomizedHyperSteiner}; additional details appear in Appendix~\ref{app:exp_det}.

\vspace{-3pt}
\subsection{Performance Across Data Distributions}
\label{sec:comp}

We evaluate RHS across three regimes reflecting typical hyperbolic embedding scenarios (Figure~\ref{fig:AcrossDataDistributionsPlot}). For each cardinality $|P|$, results are averaged over 10 independent samplings with standard deviation reported. Runtime and detailed tables appear in Appendix~\ref{app:exp}.

\textbf{Centered Gaussian.}  
This regime models the initialization phase of hyperbolic embeddings, where points concentrate near the origin \citep{nickel2017poincare, chami2020trees}. Terminals are sampled from a pseudo-hyperbolic Gaussian $\mathcal{G}(0,0.5)$ by mapping an isotropic Euclidean Gaussian from the tangent space via the exponential map \citep{Nagano2019wrapped}, yielding moderate spread around the disk center.

\textbf{Mixture of Gaussians Near Boundary.} To capture the phenomenon where tree-embedded points concentrate at the boundary in specialized groups \citep{nickel_learning_2018, Klimovskaia2020SingleCellPoincareMap, kleinberg2007geographic, chami2020trees}, we construct ten pseudo-hyperbolic Gaussians $\mathcal{G}_k(\mu_{10,k},0.1)$ centered at a regular $10$-gon near the boundary: $\mu_{10,k} = (1-10^{-10}) e^{2 i \pi k / 10}$ for $k=1,\ldots,10$, producing well-separated clusters near the boundary.

\textbf{Real Biological Data.}  
We validate generalization on the \emph{Planaria} single-cell RNA-sequencing data \citep{Plass2018_planaria_SingleCellData}, containing 21,612 cells from \emph{S.~mediterranea} undergoing hierarchical differentiation. Following PCA dimensionality reduction to 50 components, cells are embedded into two-dimensional hyperbolic space using the method of \citet{Klimovskaia2020SingleCellPoincareMap}, which preserves the biological hierarchy.

\textbf{Results.}  

Figure~\ref{fig:AcrossDataDistributionsPlot} demonstrates that RHS consistently outperforms all baselines. In the centered Gaussian regime (Figure~\ref{fig:row1}), the dense central distribution yields Euclidean behavior as $|P|$ increases due to hyperbolic space's local flatness, producing modest RHS improvements over HS ($\sim$1\%) comparable to 2D Euclidean stochastic DT-based heuristics \citep{Laarhoven2011Randomized}. Randomization becomes critical in the boundary mixture (Figure~\ref{fig:row2}): RHS maintains substantial reductions as sample size grows, while HS provides only moderate gains and NJ deteriorates to negative performance. On real biological data (Figure~\ref{fig:row3}), RHS sustains stable improvements over HS, whereas NJ exhibits inconsistent behavior. Notably, despite being a common computational biology baseline, NJ fails to respect the \textit{principle of maximum parsimony}. The marked performance advantage for boundary-concentrated clusters (Figure~\ref{fig:row2}) motivates a deeper investigation of the relationship between boundary proximity and tree-length reduction.

\begin{figure*}[t!]
\centering
\begin{subfigure}[t]{0.5\linewidth}
    \centering
    \includegraphics[width=\linewidth]{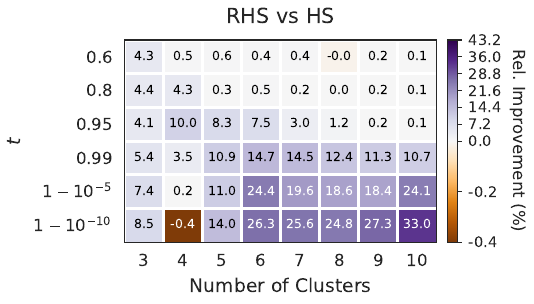}
\end{subfigure}%
\hfill
\begin{subfigure}[t]{0.5\linewidth}
    \centering
    \includegraphics[width=\linewidth]{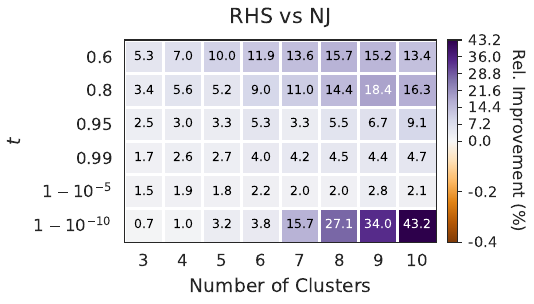}
\end{subfigure}%
\caption{
Convergence analysis for mixtures of $\mathcal{G}(\mu_{d, k}(t), 0.1)$, $k \in \{1, \ldots, d\}$ with $d \in \{3, \ldots, 10\}$ and varying radial parameter $t$, sampling 20 points per Gaussian. Values show percentage reduction in tree-length of RHS vs.\ HS (left) and NJ (right) as points approach the boundary.}
\label{fig:ConvergenceAnalysis}
\end{figure*}
\subsection{Near Boundary Performance Analysis}
\label{sec:boundary}

\begin{figure}[h]
    \centering
    \includegraphics[width=0.9\linewidth]{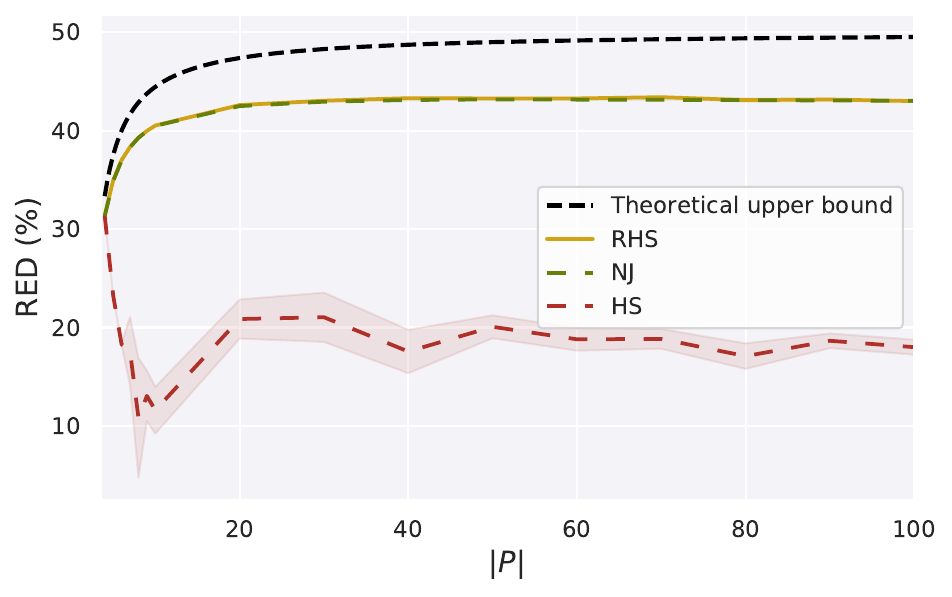}
    \caption{Tree-length reduction over the MST with $|P|$ terminals sampled near the boundary ($t=1-10^{-10}$). RHS (orange) and NJ (green) approach the theoretical bound, while HS (red) lags significantly.}
    \label{fig:exp3}
    \vspace{-10pt}
\end{figure}
The substantial improvements for boundary-concentrated clusters raise two fundamental questions: (i) how closely can heuristic methods approach the theoretical upper bound on reduction, and (ii) at what radial distance do these gains materialize? We address both through controlled experiments that systematically vary boundary proximity.

\textbf{Approaching the Theoretical Limit.}
We first evaluate performance in the extremal configuration that maximizes tree-length reduction. Following \citet{innami_steiner_2006}, the optimal arrangement for $|P|$ terminals is a regular $|P|$-gon positioned at the boundary. For $k=1,\dots,|P|$, we sample terminal $k$ from $\mathcal{G}(\mu_{|P|,k}(t),0.1)$, where $\mu_{|P|,k}(t) = t\,e^{2\pi i k / |P|}$. Setting $t=1-10^{-10}$ (the largest numerically stable value), the theoretical upper bound on reduction is $ \frac{|P|}{2(|P|-1)}$,
which converges to $50\%$ as $|P|\to\infty$. Results for each $|P|$ are averaged over three independent runs.

As we see in Figure~\ref{fig:exp3}, RHS and NJ both track this theoretical bound closely, while HS consistently underperforms. RHS achieves a maximum average reduction of $43.39\%$ at $|P|=70$, compared to $43.16\%$ for NJ at $|P|=50$, whereas HS plateaus at $31.31\%$. This translates to a relative tree-length improvement exceeding $32\%$ for RHS over the original HS heuristic. Critically, HS deteriorates rapidly beyond $|P|=5$, failing to produce competitive solutions, whereas RHS remains robust and near-optimal as terminal count grows.

\textbf{Characterizing the Transition Zone.}
To identify when proximity-driven gains emerge, we sample $20$ points per cluster from mixtures of $d\in\{3,\ldots,10\}$ Gaussians $\mathcal{G}(\mu_{d,k}(t), 0.1)$, varying the radial parameter $t$ to control distance from the origin. Figure~\ref{fig:ConvergenceAnalysis} reveals a sharp performance transition. Comparing RHS to HS (left panel), gains are negligible for $t\leq 0.8$ but increase dramatically at $t\geq 0.95$, exceeding $30\%$ improvement with $d=10$ clusters at the extreme boundary. This transition marks where hyperbolic curvature dominates local Euclidean approximations. While NJ performed comparably to RHS in the single point-per-cluster configuration of Figure~\ref{fig:exp3}, the cluster-based setting reveals a significant performance gap, with RHS achieving over $40\%$ improvement in tree-length.

These results demonstrate that RHS's stochastic construction bridges heuristic and optimal performance. As established in Section~\ref{sec:bg}, the hyperbolic plane admits maximum possible MST reduction among all complete boundaryless surfaces. RHS exploits this geometric property to attain near-optimal reductions, representing a significant advance over existing heuristics and revealing a key insight: hyperbolic geometry provides computational advantages for the Steiner tree problem unattainable in Euclidean spaces.

\section{Theoretical Limitations of Stochastic Topology Search}

Our method shares a key limitation with prior Euclidean stochastic Steiner tree heuristics \citep{Beasley1994DTHeuristicSMT, Laarhoven2011Randomized, yang2006hybrid}: for general point clouds, it does not provide a theoretical guarantee of finding the globally optimal topology. While Theorem~\ref{thm:uniqueness} shows that, for a fixed topology, the Steiner point configuration is uniquely determined, and thus the global refinement step showcased in Section~\ref{sec:glob_ref} is well posed, the topology itself remains unknown and must be explored heuristically.

\begin{figure*}[t!]
    \centering
    \begin{subfigure}[t]{0.15\textwidth}
        \centering
        \includegraphics[width=\linewidth]{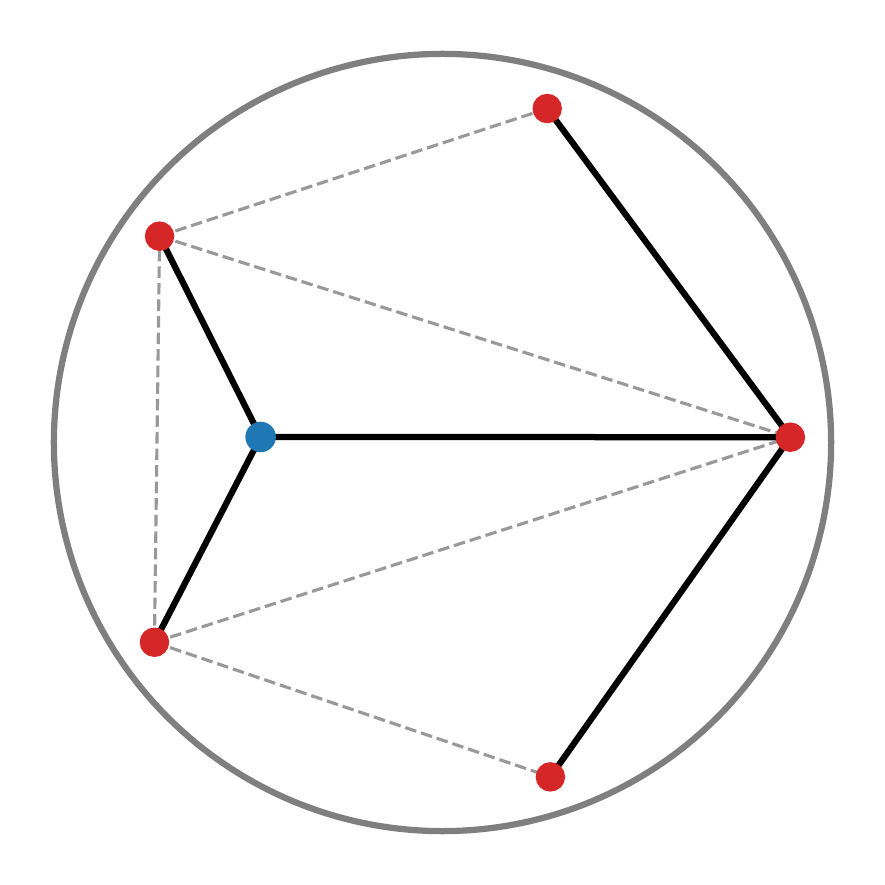}
        \caption*{\footnotesize HS: $|P|=5$}
    \end{subfigure}
    \hfill
    \begin{subfigure}[t]{0.15\textwidth}
        \centering
        \includegraphics[width=\linewidth]{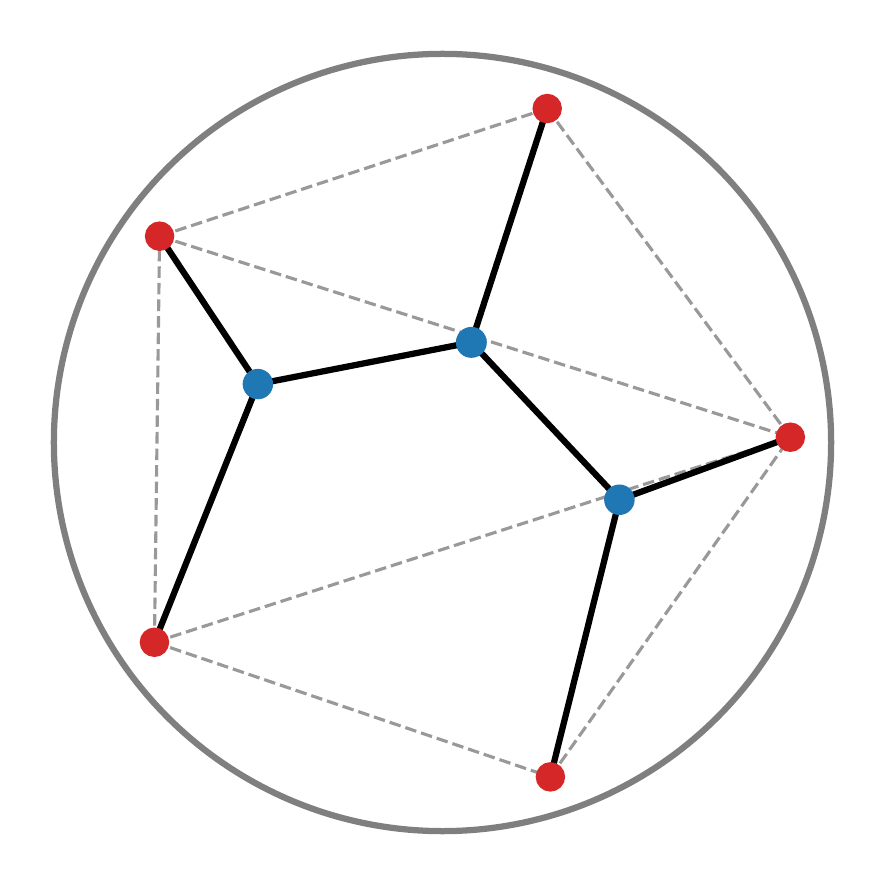}
        \caption*{\footnotesize RHS: $|P|=5$}
    \end{subfigure}
    \hfill
    \begin{subfigure}[t]{0.15\textwidth}
        \centering
        \includegraphics[width=\linewidth]{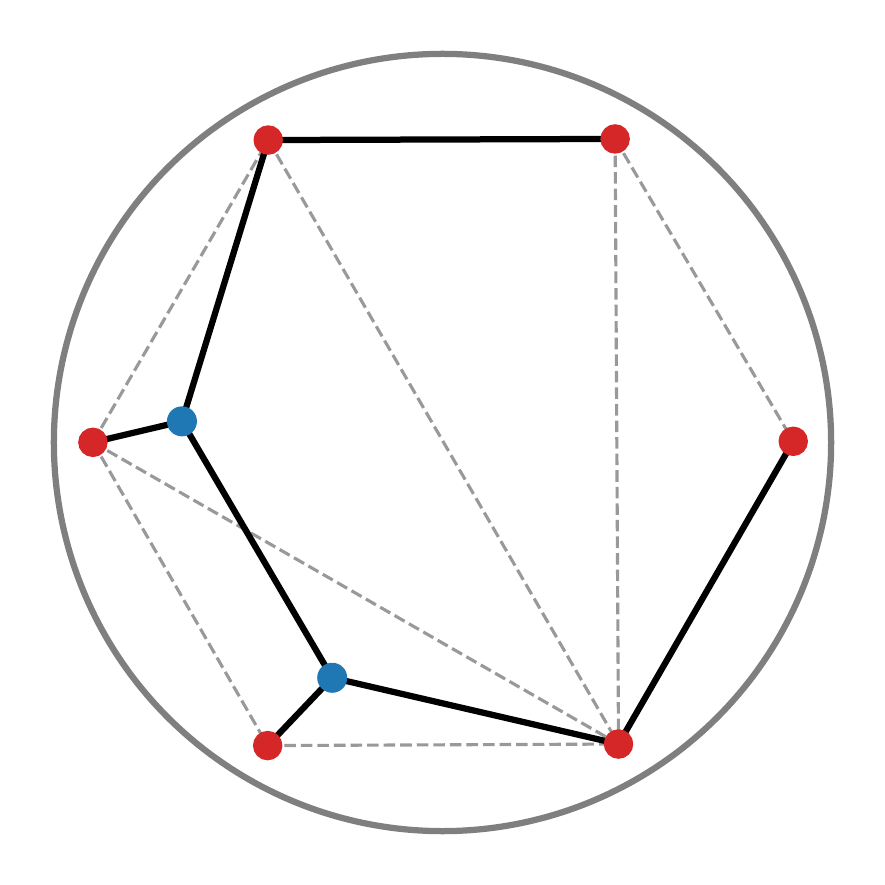}
        \caption*{\footnotesize HS: $|P|=6$}
    \end{subfigure}
    \hfill
    \begin{subfigure}[t]{0.15\textwidth}
        \centering
        \includegraphics[width=\linewidth]{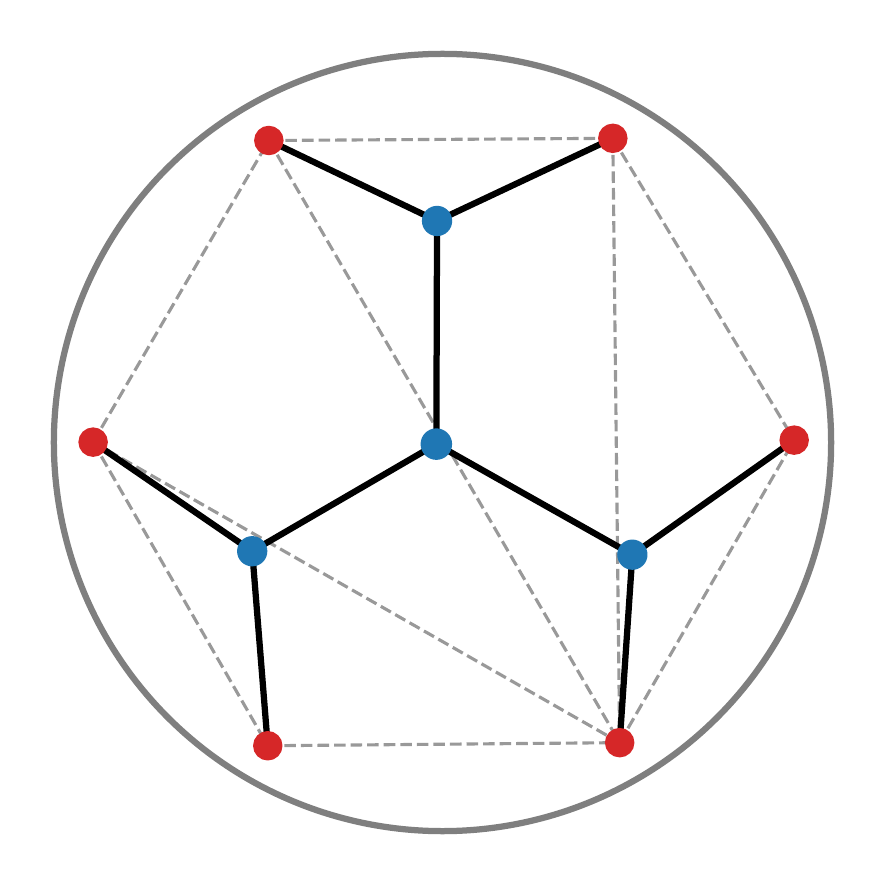}
        \caption*{\footnotesize RHS: $|P|=6$}
    \end{subfigure}
    \hfill
    \begin{subfigure}[t]{0.15\textwidth}
        \centering
        \includegraphics[width=\linewidth]{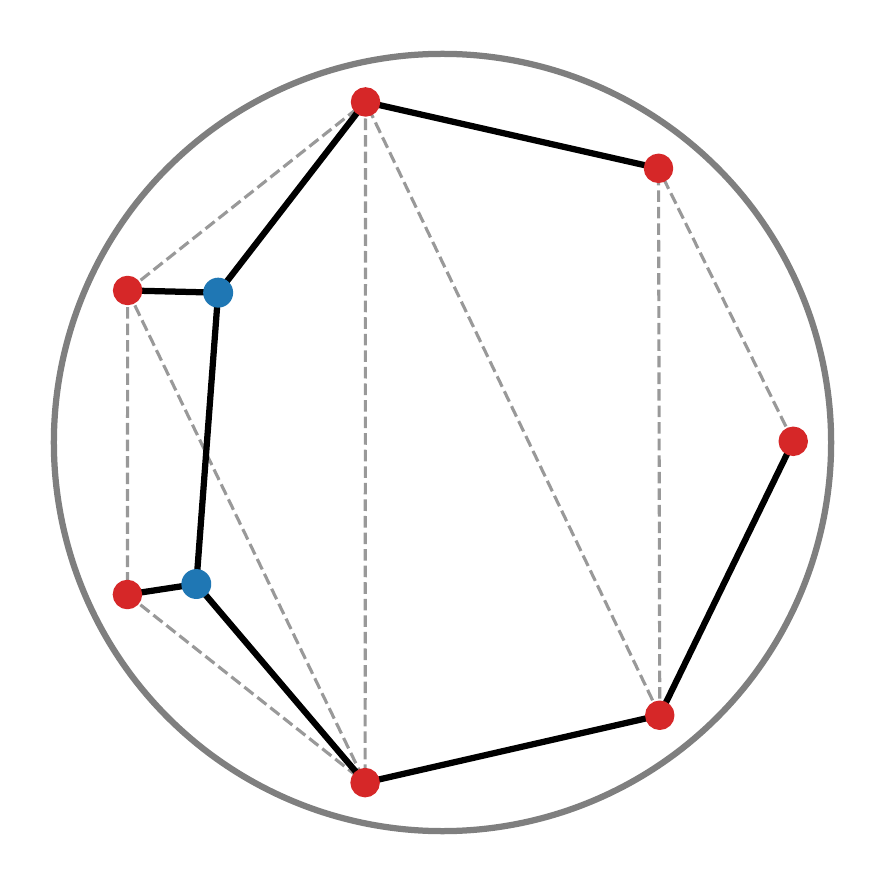}
        \caption*{\footnotesize HS: $|P|=7$}
    \end{subfigure}
    \hfill
    \begin{subfigure}[t]{0.15\textwidth}
        \centering
        \includegraphics[width=\linewidth]{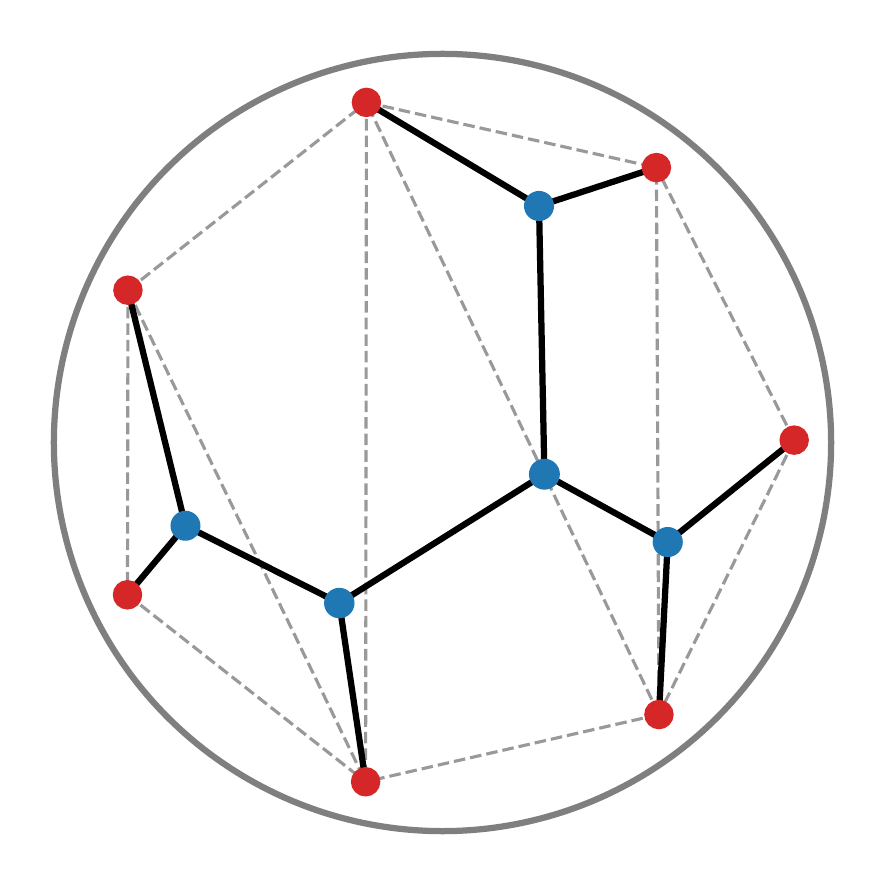}
        \caption*{\footnotesize RHS: $|P|=7$}
    \end{subfigure}
    \caption{Heuristic Steiner Minimum Trees computed by the deterministic HyperSteiner (HS) and the proposed Randomized HyperSteiner (RHS) on regular $|P|$-gons placed near the boundary ($t=0.9$). Red ({\color{RedLight}$\newmoon$}) denotes terminals, blue ({\color{NavyBlue}$\newmoon$}) denotes Steiner points, and dashed lines correspond to the auxiliary hyperbolic DT. Our method consistently produces more hierarchical tree structures than its deterministic counterpart.}
    \label{fig:plots-poly}
\end{figure*}

Accordingly, RHS relies on stochastic exploration over candidate topologies while retaining the best solution found so far. This makes standard convergence analysis infeasible. In particular, changes in topology can produce discontinuous changes in tree length, which makes it difficult to obtain a straightforward characterization of convergence or theoretical computational complexity, as in other non-deterministic Steiner tree methods \citep{Beasley1994DTHeuristicSMT, Laarhoven2011Randomized}.

Despite these theoretical limitations, the results in Section~\ref{sec:exp} and Appendix~\ref{app:exp} show that the proposed pipeline consistently improves over existing baselines across a range of synthetic and real hyperbolic datasets.

\section{Discussion and Future Work}
\label{sec:diss}

We introduced \textbf{Randomized HyperSteiner} (RHS), a stochastic Delaunay triangulation heuristic that overcomes the fundamental limitation of deterministic approaches to constructing hyperbolic Steiner Minimal Trees. By incorporating randomness into the expansion process and combining it with Riemannian gradient descent optimization, RHS avoids the local optima that trap existing methods. Moreover, as shown in Figure~\ref{fig:plots-poly}, the myopic behavior of the original HyperSteiner prevents it from generating complex hierarchical trees, whereas our proposed RHS, thanks to our iterative exploration procedure, is able to construct richer hierarchical structures with shorter total tree length.

Empirically, RHS consistently outperforms existing baselines across diverse synthetic and real data distributions, with particularly strong performance in near-boundary configurations where hyperbolic geometry's advantages are most pronounced. However, as is common with stochastic heuristics for the SMT problem \citep{Laarhoven2011Randomized}, this improved solution quality comes at the cost of increased computational time compared to deterministic approaches like the vanilla HyperSteiner. This presents practitioners with a clear trade-off: RHS should be chosen when solution quality is paramount, while the vanilla HyperSteiner remains suitable for applications requiring fast approximations. We refer the reader to Appendix~\ref{app:scala} for further discussion on the computational trade-off.

Despite these advances, our results reveal that we have not yet achieved the full potential of hyperbolic geometry for tree length minimization. A gap of approximately 7\% remains between our best empirical results and the theoretical upper bound. A promising direction for future work involves developing Polynomial-Time Approximation Schemes (PTAS) using quadtrees instead of Delaunay triangulations \citep{Arora1998ProofVerificationHardness, Arora1998PTASTSP, Mitchell1999PTASTSP}. PTAS approaches would allow users to explicitly control the performance-time trade-off by adjusting the approximation parameter. While the construction of quadtrees in hyperbolic space has been shown to be theoretically feasible \citep{kisfaludi2025near}, the practical implementation of these polynomial-time approximation algorithms in hyperbolic geometry remains an open challenge that could finally bridge the gap to optimal performance.

\section*{Acknowledgements}
This work has been supported by the Swedish Research Council, the Knut and Alice Wallenberg Foundation, the European Research Council (ERC AdG BIRD), and the Swedish Foundation for Strategic Research (SSF BOS). Alejandro García Castellanos is funded by the Hybrid Intelligence Center, a 10-year programme funded through the research programme Gravitation which is (partly) financed by the Dutch Research Council (NWO). This publication is part of the project SIGN with file number VI.Vidi.233.220 of the research programme Vidi which is (partly) financed by the Dutch Research Council (NWO) under the grant \url{https://doi.org/10.61686/PKQGZ71565}. 
\clearpage
\newpage
\bibliography{references}

\newpage

\clearpage
\appendix
\thispagestyle{empty}

\clearpage
\onecolumn
\aistatstitle{Randomized HyperSteiner: Supplementary Materials}

\section{Additional Details of Randomized HyperSteiner}
\label{app:heuristic}

This appendix provides further details on the design and implementation of the proposed Randomized HyperSteiner (RHS) method. 

\subsection{Expansion Phase}
\label{app:exp_phase}
\begin{wrapfigure}{r}{0.4\linewidth}
\vspace{-10pt}
  \centering
  \includegraphics[width=\linewidth]{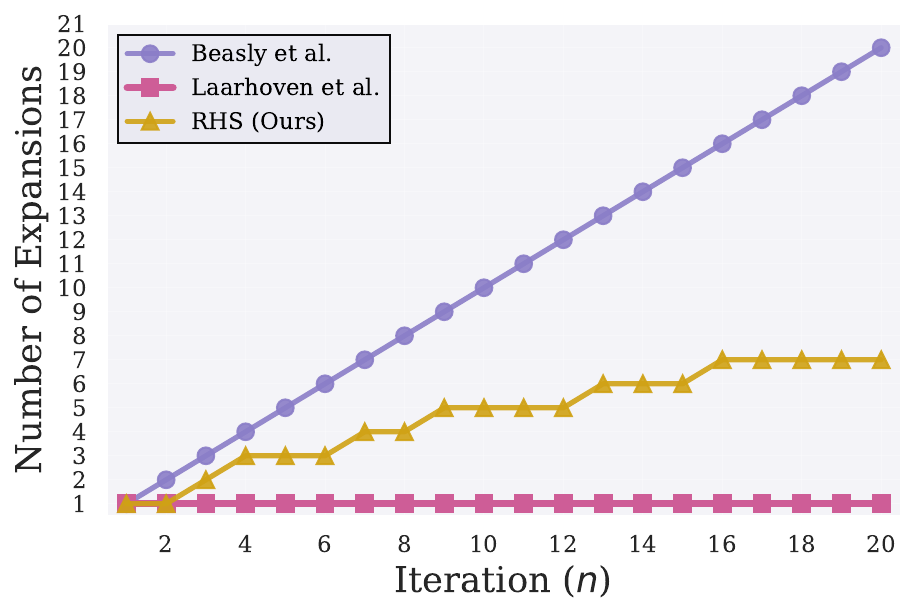}
  \caption{Comparison of expansion strategies.}
  \label{fig:exps}
  \vspace{-14pt}
\end{wrapfigure}
In the expansion phase, we iteratively construct the Delaunay triangulation (DT) of the union of the terminal set and the current Steiner candidates $S$. Following \citet{Laarhoven2011Randomized}, we begin by sampling an insertion probability $p \sim \mathcal{U}[0.3,0.6]$. Then, for each triangle in $\text{DT}(P \cup S)$, we sample from a Bernoulli distribution with success probability $p$ to decide whether to insert its hyperbolic barycenter as a new Steiner candidate. For three points $(u, v, w)$ in the Klein disk, the hyperbolic barycenter is given by
\[
m_{uvw} = \frac{\gamma_u u + \gamma_v v + \gamma_w w}{\gamma_u + \gamma_v + \gamma_w}, \quad \gamma_u = \frac{1}{\sqrt{1-\|u\|^2}},
\]
where $\gamma_u$ is the Lorentz factor \citep{Ungar2009}.

At iteration $n$, we perform $\lfloor 2\sqrt{n} - 1 \rfloor$ expansion steps (see Algorithm~\ref{alg:rad_hyper}). As discussed in Section~\ref{sec:meth}, this schedule provides finer control of the search process and allows revisiting promising configurations. By contrast, \citet{Beasley1994DTHeuristicSMT} perform exactly $n$ expansions per iteration, which scales linearly and may lead to excessive growth in later stages (see Figure~\ref{fig:exps}).

Detailed ablations regarding the insertion probabilities and the expansion scheduler can be found in Section~\ref{app:ins_prob} and Section~\ref{app:exp_sche}.

\subsection{Verification of Degree and Angle Conditions}

Algorithms~\ref{alg:red} and~\ref{alg:edge-in} enforce, respectively, the degree and angle conditions for candidate Steiner configurations. 

\paragraph{Degree condition.} Algorithm~\ref{alg:red} locally verifies degree constraints by evaluating candidate full Steiner trees on three- and four-point neighborhoods. For the three-point case, the system of equations in \eqref{eq:system} can be solved using standard numerical solvers. For the four-point case, we follow the iterative approximation as done in \citet{garcia2025hypersteiner}. Moreover, Algorithm~\ref{alg:red} leverages the fact that, given optimal Steiner points $S$, the resulting Steiner tree is the minimum spanning tree (MST) on $P \cup S$. Since in hyperbolic space the MST is always a subgraph of the DT \citep{preparata2012computational}, we can compute the MST via the DT in $\mathcal{O}(|P \cup S'| \log |P \cup S'|)$ time, avoiding redundant computation \citep{smith1981n}.

\paragraph{Angle condition.} Algorithm~\ref{alg:edge-in} checks the $120^\circ$ angle condition in constant time using the hyperbolic cosine rule. Specifically, the angle between edges $(x_i, x_j)$ and $(x_j, x_l)$ is given by

\[
\theta = \arccos\!\left(\frac{\cosh\!\left(d_{\mathbb{K}}(x_j, x_i)\right)\cosh\!\left(d_{\mathbb{K}}(x_j, x_l)\right) - \cosh\!\left(d_{\mathbb{K}}(x_i, x_l)\right)}{\sinh\!\left(d_{\mathbb{K}}(x_j, x_i)\right)\sinh\!\left(d_{\mathbb{K}}(x_j, x_l)\right)}\right).
\]

\begin{algorithm}[ht!]
\caption{Reduction via Degree Condition Verification --- Adapted from \cite{Beasley1994DTHeuristicSMT}}
\label{alg:red}
\begin{algorithmic}

\vspace{1mm}
\STATE \textbf{Input}: Terminal set $P\subset\mathbb{K}^2$, Steiner set $S\subset\mathbb{K}^2$
\vspace{1mm}

\STATE $S' \leftarrow S$
\STATE $T \leftarrow \text{MST}(P \cup S')$

\REPEAT
    \FOR{each $s \in S'$}
        \STATE $d \leftarrow \deg_T(s)$
        \IF{$d \leq 2$ \textbf{or} $d \geq 5$}
            \STATE $S' \leftarrow S' \setminus \{s\}$
        \ELSIF{$d = 3$}
            \STATE $N \leftarrow$ neighbors of $s$ in $T$
            \STATE $s^* \leftarrow$ optimal Steiner point for triangle $N$
            \IF{$s^*$ exists}
                \STATE Replace $s$ with $s^*$ in $S'$
            \ELSE
                \STATE $S' \leftarrow S' \setminus \{s\}$
            \ENDIF
        \ENDIF
    \ENDFOR
    \STATE $T \leftarrow \text{MST}(P \cup S')$
\UNTIL{no changes in $S'$}

\FOR{each $s \in S'$}
    \IF{$\deg_T(s) = 4$}
        \STATE $N \leftarrow$ neighbors of $s$ in $T$
        \STATE $Q, (s_1, s_2) \leftarrow$ optimal Full Steiner tree for 4-point set $N$, and Steiner points
        \STATE $S' \leftarrow (S' \setminus \{s\}) \cup \{s_1, s_2\}$
        \STATE Substitute the subgraph of the neighborhood of $s$ on $T$ by $Q$
    \ENDIF
\ENDFOR

\vspace{1mm}
\STATE \textbf{Output}: Reduced Steiner set $S'$, tree $T$ on $P \cup S'$
\vspace{1mm}
\end{algorithmic}
\end{algorithm}

\begin{algorithm}[ht!]
\caption{Expansion via Angle Condition Verification --- Adapted from \cite{thompson1973method}}
\label{alg:edge-in}
\begin{algorithmic}

\vspace{1mm}
\STATE \textbf{Input}: Terminal point set $P\subset\mathbb{K}^2$, Steiner point set $S\subset\mathbb{K}^2$, tree $T$ on $P \cup S$
\vspace{1mm}

\STATE Let $E$ be the set of edges in $T$.
\STATE Set $T' \leftarrow T$, $S' \leftarrow S$, and edge list $\leftarrow \emptyset$.

\FOR{$i = 1, \ldots, |X \cup S|$}
    \FOR{$j$ such that $(x_i, x_j) \in E$}
        \FOR{$l$ such that $(x_j, x_l) \in E$}
            \IF{$(x_j, x_l)$ meets $(x_i, x_j)$ at angle less than $120^{\circ}$}
                \STATE edge list $\leftarrow \text{edge list}\cup (x_j, x_l)$.
            \ENDIF
        \ENDFOR
        \IF{edge list $\neq \emptyset$}
            \STATE From edge list, select $(x_j, x_k)$ which minimizes the angle with $(x_i, x_j)$.
            \STATE $s^* \leftarrow$ optimal Steiner point for triangle $\triangle(x_i, x_j, x_k)$
            \IF{$s^*$ exists}
            \STATE Let $S' \leftarrow S \cup s^*$.
            \STATE Remove edges $(x_i, x_j)$ and $(x_j, x_k)$ from $T'$.
            \STATE Add edges $(x_i, s^*)$, $(x_j, s^*)$, and $(x_k, s^*)$ to $T'$.
            \ENDIF
         \STATE Set edge list $\leftarrow \emptyset$.
        \ENDIF
    \ENDFOR
\ENDFOR

\vspace{1mm}
\STATE \textbf{Output}: Expanded Steiner set $S'$, tree $T'$ on $P \cup S'$
\vspace{1mm}

\end{algorithmic}
\end{algorithm}

\vspace{50pt}
\subsection{Global Refinement via Riemannian Gradient Descent}
\label{app:reim_desc}
As discussed by \citet{lezcano2019trivializations}, evaluating the exponential map can be computationally expensive, motivating the use of retractions as efficient first-order approximations.

\begin{defn}[Retraction]\label{def:retraction}
A differentiable map $r : T\mathcal{M} \to \mathcal{M}$ is a \emph{retraction} if, for every $p \in \mathcal{M}$, the map $r_p : T_p\mathcal{M}\to \mathcal{M}$ satisfies:
\begin{enumerate}
    \item $r_p(0) = p$ (identity at the origin), and
    \item $(\mathrm{d}\, r_p)_0 = \mathrm{Id}$ (first-order equivalence to the exponential map).
\end{enumerate}
\end{defn}

Retractions provide efficient approximations of geodesic flow near the origin. For our problem, the update rule becomes
\[
S^{(t+1)} = r_{S^{(t)}}\left(-\eta \nabla L\left(S^{(t)}\right)\right),
\]
where $L$ is the tree length function and $\eta$ is the learning rate. As shown by \citet{boumal2016global}, gradient descent with retractions preserves convergence guarantees of exponential map-based methods while being significantly more efficient. 

To further reduce runtime, we apply early stopping upon convergence (see Appendix~\ref{app:exp}). While we use a fixed learning rate, more advanced line-search strategies (e.g., Armijo backtracking \citep{absil2008optimization}) could further improve stability and convergence speed.

\subsection{Exploration Policy}
Following \citet{Laarhoven2011Randomized}, we set the maximum number of iterations to $N^* = \lfloor \sqrt{|P|}\rfloor$. This choice balances exploration and efficiency, enabling sufficient search depth to discover improved topologies while avoiding wasted computation. The algorithm maintains the best solution found ($T^*$) and continues exploring until either (i) a better solution is identified ($L(T) < L(T^*)$), triggering a restart from the new checkpoint, or (ii) the iteration budget $N^*$ is exhausted. Figure~\ref{fig:explore} illustrates this policy.

\begin{figure}[ht!]
    \centering
    \includegraphics[width=0.95\linewidth]{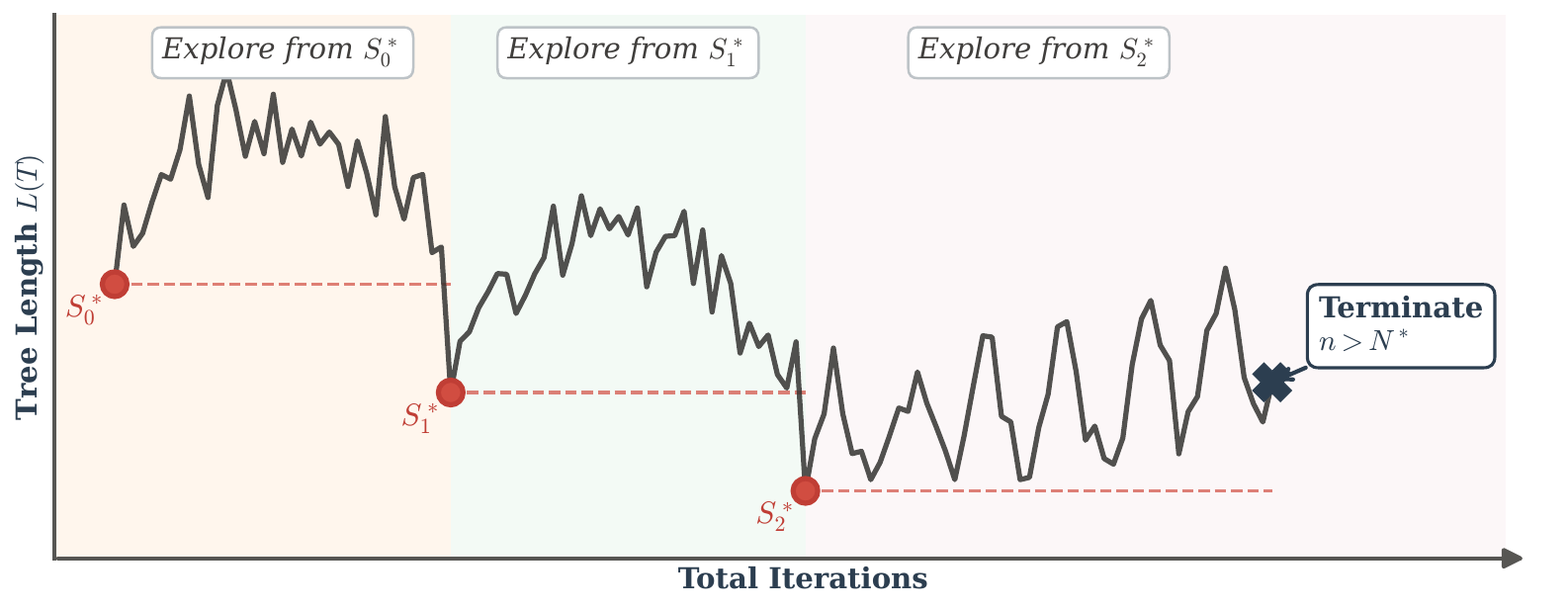}
    \caption{Exploration policy of RHS. The search restarts from a new checkpoint when a better configuration is discovered.}
    \label{fig:explore}
\end{figure}

\section{Review on Original HyperSteiner}
\label{app:og_hyp}

In this section, we review the original HyperSteiner algorithm presented in \cite{garcia2025hypersteiner}. As mentioned in Section~\ref{sec:rel}, Algorithm~\ref{alg:og_hyper} is based on the Smith-Lee-Liebman (SLL) algorithm \citep{smith1981n} with a minor modification from \cite{zachariasen_concatenation-based_1999}.

\begin{algorithm}[ht]
\caption{Original HyperSteiner}
\label{alg:og_hyper}
\begin{algorithmic}
\vspace{1mm}
\STATE \textbf{Input}: Terminal point set $P\subset\mathbb{K}^2$
\vspace{1mm}

\STATE $\triangleright\, \textsc{\textbf{Deterministic MST-Based Expansion}}$\vspace{-5pt}

    \begin{tcolorbox}[colback=lightcyan, colframe=cyan!70!black, boxrule=0.5pt]
\begin{enumerate}
    \item Construct the Delaunay triangulation, $\text{DT}(P)$.
    \item Construct $\text{MST}(P)$ (Kruskal algorithm) and simultaneously build a priority queue as follows:
    \begin{enumerate}[label*=\arabic*.]
        \item Mark all the triangles $\sigma \in \text{DT}(P)$ containing two edges of $\text{MST}(P)$ and admitting an FST. 
        \item Place the FSTs of marked triangles $\sigma$ in a queue $Q$ prioritized on $\rho(\sigma)$ (smaller first).
    \end{enumerate}
    \end{enumerate}
    \end{tcolorbox}

    \STATE $\triangleright\, \textsc{\textbf{Heuristic Construction}}$\vspace{-5pt}
\begin{tcolorbox}[colback=lightyellow, colframe=orange!70!black, boxrule=0.5pt]

    \begin{enumerate}
     \setcounter{enumi}{2}

    \item Add the FST of the $4$-terminal subsets:
    \begin{enumerate}[label*=\arabic*.]
        \item For each marked triangle $\sigma$, find its adjacent triangles $\sigma'$ such that $\sigma$ and $\sigma'$ contain three edges of the $\text{MST}(P)$.
        \item Compute the FST $\sigma \cup \sigma'$ for each of the two possible topologies and add the minimal one to $Q$. 
    \end{enumerate}
    \item Convert $Q$ to an ordered list and append to it the edges of $\text{MST}(P)$, sorted in non-decreasing order.
    \item Let $T$ be an empty tree. An FST in $Q$ is added to $T$ if it does not create a cycle (greedy concatenation).  
\end{enumerate}

\end{tcolorbox}
    
\vspace{1mm}
\STATE \textbf{Output}: A heuristic SMT $T$ 
\vspace{1mm}
\end{algorithmic}
\end{algorithm}

The SLL algorithm uses the MST as an initial approximation of the SMT. This approach relies on the Gilbert-Pollak conjecture that the Steiner ratio satisfies $\rho(P) \geq \sqrt{3}/2 \approx 0.866$ for any $P\subset \mathbb{R}^2$, which implies ${L(\text{MST}(P)) \leq 2/\sqrt{3} \cdot L(\text{SMT}(P))}$\footnote{The Gilbert-Pollak conjecture remains open \citep{ivanov2012steiner}. The best known lower bound is $\sim 0.824$ \citep{news_steiner_ratio}.}. Similarly, the original HyperSteiner is based on the fact that, as described in Section~\ref {sec:bg},  on hyperbolic space the Steiner ratio satisfies $\rho(P) \geq 1/2$ for any $P\subset \mathbb{K}^2$ \citep{innami_steiner_2006}.

The original HyperSteiner algorithm replaces MST edges connecting 3-terminal and 4-terminal subsets with their corresponding full Steiner trees (FSTs). Since the MST is a subgraph of the Delaunay triangulation \citep{preparata2012computational}, the algorithm restricts its search to vertices of Delaunay triangles.

This algorithm follows the standard three-phase DT-based framework from Section~\ref{sec:rel}: expansion, heuristic construction, and decision policy. The expansion phase builds a priority queue of 3-terminal FSTs based on Steiner ratios. The construction phase adds 4-terminal FSTs from adjacent triangle pairs. The decision policy uses greedy concatenation, accepting all improvements that do not create cycles.

As shown in Section~\ref{sec:exp}, the original HyperSteiner can be myopic due to its reliance on the MST as a guiding structure, leading to locally optimal but globally suboptimal configurations. However, this approach achieves fast computational complexity of $\mathcal{O}(|P| \log |P|)$.

\clearpage
\section{Uniqueness of Relatively Minimal Trees in Hyperbolic Space}
\label{app:uniq}

\begin{lemm}
\label{lem:segment_isoptic}
Let $x,y$ be two distinct points strictly inside the Klein disk and let $z\neq z'$ be two distinct points. Parametrize the geodesic segment
\[
s(t) = (1-t)z + tz', \qquad t\in[0,1].
\]
Then the entire segment $[z,z']$ lies in the $120^\circ$--isoptic of $(x,y)$ if and only if $z=z'$. In particular, if $z\neq z'$ the whole segment cannot be contained in the isoptic.
\end{lemm}

\begin{proof}
The automorphism group of the disk acts transitively on ordered pairs of distinct interior points and consists of hyperbolic isometries. Thus there exists a disk automorphism $\Phi$ with
\[
\Phi(z)=(0,0), \qquad \Phi(z')=(r,0) \text{ with } r\in (0,1).
\]
Because $\Phi$ is an isometry, the hyperbolic distance is preserved and
\[
0\neq d_{\mathbb K}(z,z') = d_{\mathbb K}((0,0),(r,0)) = 2\tanh^{-1}(r),
\]
so $r = \tanh(\tfrac12 d_{\mathbb K}(z,z'))\in (0,1)$. In particular $r$ is uniquely determined by the hyperbolic length of $[z,z']$. The isoptic property is preserved by $\Phi$, so it suffices to prove the statement for the normalized configuration with $\hat{z}:=\Phi(z)=(0,0)$ and $\hat{z}':=\Phi(z')=(r,0)$.

Write $\hat{x}:=\Phi(x)=(a,b)$ and $\hat{y}:=\Phi(y)=(c,d)$ with $a^2+b^2<1,\;c^2+d^2<1$. Parametrize the normalized segment by
\[
s(t) = (rt,0), \qquad t\in[0,1].
\]
Let $\langle u,v\rangle=-1+u_0v_0+u_1v_1$ denote the Lorentzian inner product. Define
\[
\begin{aligned}
u(t)&=\langle \hat{x},s(t)\rangle, &
v(t)&=\langle \hat{y},s(t)\rangle, &
w(t)&=\langle s(t),s(t)\rangle,\\
L(t)&=u(t)v(t)-\langle \hat{x},\hat{y}\rangle w(t), &
p_1(t)&=u(t)^2-\langle \hat{x},\hat{x}\rangle w(t), &
p_2(t)&=v(t)^2-\langle \hat{y},\hat{y}\rangle w(t).
\end{aligned}
\]
Hence, we can rewrite the isoptic condition of Proposition~\ref{prop:isoptic}, as
\[
\varphi(s(t)) = L(t) + \tfrac12 \sqrt{p_1(t)p_2(t)} = 0.
\]
If $\varphi(s(t))\equiv 0$ on the entire segment, then squaring yields the polynomial identity
\begin{equation}
\label{eq:quartic}
    Q(t):=4L(t)^2 - p_1(t)p_2(t) \equiv 0.
\end{equation}

Since $\deg Q \le 4$, this means all coefficients of $Q(t)$ vanish.

Expanding $Q$ with a Computer Algebra System (CAS), as done in Figure~\ref{fig:coeff_code}, one finds that the coefficient of $t^4$ is
\[
E_4 = r^4\Big(3b^2d^2 + b^2 - 8bd + d^2 + 3\Big)
     = r^4\Big(3(bd-1)^2 + (b-d)^2\Big).
\]
Because $0<r<1$ we have $r^4>0$, and the parenthesis is a sum of two squares, hence nonnegative. Equality $E_4=0$ would force simultaneously
\[
(bd)r^2 = 1, \qquad b = dr^2,
\]
which implies $d^2=1/r^4>1$. This contradicts the assumption that $\hat{y}=(c,d)$ lies strictly inside the unit disk, since $c^2+d^2<1$ requires $|d|<1$. Therefore $E_4>0$ for any two interior foci $\hat{x},\hat{y}$.

Hence the quartic $Q(t)$ cannot vanish identically on the line. Thus the entire nontrivial segment $[z,z']$ cannot lie in the isoptic. The only case where the whole segment is contained is when $z=z'$, which is trivial.
\end{proof}

\begin{figure*}[t!]
    \centering
    \includegraphics[width=0.8\textwidth]{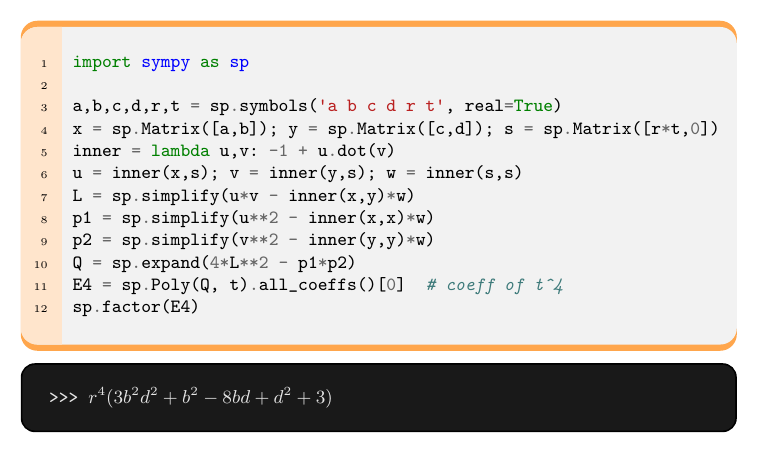}
    \vspace{-5mm}
    \caption{Simple SymPy code snippet for computing the coefficients of the quartic presented in Equation~\ref{eq:quartic}.}
    \vspace{-2mm}
    \label{fig:coeff_code}
\end{figure*}

\begin{prop}[\cite{Bridson1999}]
\label{prop:convexity_hadamard}
Let $(\mathcal{H}, g)$ be a Hadamard manifold, i.e., a complete, simply connected Riemannian manifold with sectional curvature $K < 0$ (this includes hyperbolic spaces). For any two geodesics $\gamma, \eta: [0,1] \to \mathcal{H}$, the distance function $t \mapsto d(\gamma(t), \eta(t))$ is convex on $[0,1]$. That is,
\begin{equation}
d(\gamma(t), \eta(t)) \leq (1-t) d(\gamma(0), \eta(0)) + t d(\gamma(1), \eta(1))
\end{equation}
for all $t \in [0,1]$.
\end{prop}

\paragraph{Notation.}  
Let $P_1,\dots,P_n$ be fixed terminal points in the Klein disk model of hyperbolic space.  
Fix a tree topology with adjacency weights $a_{ij}\in\{0,1\}$ and $s$ Steiner points.  
A \emph{configuration} of vertices is a tuple
\[
V = (V_1,\dots,V_{n+s}) \quad\text{with}\quad V_i = P_i \ \text{for } 1\le i\le n,\;\; V_k = S_k \ \text{for } k>n.
\]
The total tree length is
\[
L(V) = \sum_{i<j} a_{ij}\, d_{\mathbb K}(V_i,V_j),
\]
where $d_{\mathbb K}(\cdot,\cdot)$ is the Klein distance.  

Given two configurations $V(0)$ and $V(1)$ with the same terminals but different Steiner points, define the interpolated configuration $V(t)$ for $t\in[0,1]$ by
\[
V_i(t) =
\begin{cases}
P_i, & i\le n \quad (\text{terminals fixed}), \\
(1-t)S_i + tS_i', & i>n \quad (\text{Steiner interpolation}).
\end{cases}
\]
Since geodesics in the Klein disk are Euclidean chords, each $t\mapsto V_i(t)$ is a hyperbolic geodesic (possibly constant if $i\le n$).  
The length of the interpolated tree is
\[
L(t) = L(V(t)) = \sum_{i<j} a_{ij}\, d_{\mathbb K}(V_i(t),V_j(t)).
\]

\newtheorem*{thUniq}{Theorem~\ref{thm:uniqueness} (Restated)}

\begin{thUniq}
Given a fixed tree topology and fixed terminals $P_1,\dots,P_n$ in the Klein disk model, the minimizing configuration of Steiner points $S_{n+1},\dots,S_{n+s}$ is unique.
\end{thUniq}

\begin{proof} This proof will closely follow the steps of the Euclidean counter part presented in \cite{Gilbert1968SMT}, but showing that can be adapted for the Klein disk.

\noindent\textbf{Step 1. Convexity of $L(t)$.}
Fix a pair $(i,j)$ with $a_{ij}=1$.  
Both $t\mapsto V_i(t)$ and $t\mapsto V_j(t)$ are geodesics (constant if terminal, chordal if Steiner).  
By Proposition~\ref{prop:convexity_hadamard}, the function
\[
t \;\mapsto\; d_{\mathbb K}\big(V_i(t),V_j(t)\big)
\]
is convex. Explicitly,
\begin{equation}\label{eq:pairwise}
d_{\mathbb K}\big(V_i(t),V_j(t)\big) \le (1-t)d_{\mathbb K}\big(V_i(0),V_j(0)\big) + t d_{\mathbb K}\big(V_i(1),V_j(1)\big).
\end{equation}
(If one of the geodesics is constant — corresponding to a terminal — the statement still holds: the distance to a fixed point is convex along any geodesic.)

Multiplying \eqref{eq:pairwise} by $a_{ij}\ge 0$ and summing over all pairs $i<j$, we obtain
\[
L(t) = \sum_{i<j} a_{ij}\, d_{\mathbb K}(V_i(t),V_j(t)) \;\le\;
(1-t)\sum_{i<j} a_{ij}\, d_{\mathbb K}(V_i(0),V_j(0)) + t\sum_{i<j} a_{ij}\, d_{\mathbb K}(V_i(1),V_j(1)).
\]
That is,
\begin{equation}\label{eq:convexity_tree}
L(t) \le (1-t)L(0) + tL(1).
\end{equation}
Hence $L(t)$ is convex in $t$.

\noindent\textbf{Step 2. Equality at minima.}
Suppose two configurations $V(0)$ and $V(1)$ both minimize the tree length, i.e.\ $L(0)=L(1)=L_{\min}$.  
Then from \eqref{eq:convexity_tree},
\[
L(t) \le (1-t)L(0)+tL(1) = L_{\min}.
\]
But since $L(t)$ is by definition the length of a feasible tree, we must also have $L(t)\ge L_{\min}$.  
Therefore $L(t)=L_{\min}$ for all $t\in[0,1]$.

\noindent\textbf{Step 3. Rigidity of Steiner points.}
Consider a Steiner point $S$ in $V(0)$ and its counterpart $S'$ in $V(1)$.  
At a minimizing configuration, incident edges meet at $120^\circ$.  
Along the interpolation $S(t)$, the Steiner point would have to remain in the $120^\circ$--isoptic of its neighbors.  
By Lemma~\ref{lem:segment_isoptic}, this is only possible if $S=S'$.  
Thus all Steiner points coincide in the two minimizing configurations.

\noindent\textbf{Step 4. Induction on the number of Steiner points.}
If $s=0$, there are no Steiner points and the tree is uniquely determined by the terminals and topology.  
Assume uniqueness holds for $s-1$ Steiner points.  
For $s$ Steiner points, Step~3 shows that any two minimal configurations agree on at least one Steiner point.  
Treating this common Steiner point as an additional fixed terminal reduces the problem to $s-1$ Steiner points, where uniqueness holds by the induction hypothesis.  
Thus the full configuration is unique.

\end{proof}

\begin{lemm}[Model-independence]
\label{lem:model_independence}
Let $\mathbb{H}^2$ denote $2$-dimensional hyperbolic space, and let $\mathcal{M}_1, \mathcal{M}_2$ be two standard models of $\mathbb{H}^2$ (e.g.\ Klein disk, Poincaré disk, upper half--plane, hyperboloid).  
Suppose the Steiner problem with fixed terminals has a unique minimizing configuration of Steiner points in $\mathcal{M}_1$.  
Then the corresponding minimizing configuration in $\mathcal{M}_2$ is also unique.
\end{lemm}

\begin{proof}
There exists a bijective isometry 
\[
\iota : \mathcal{M}_1 \to \mathcal{M}_2
\]
between the two models.  
For any configuration $V$ in $\mathcal{M}_1$, define $\iota(V)$ as the image configuration in $\mathcal{M}_2$.  
Since $\iota$ is an isometry, we have
\[
d_{\mathcal{M}_2}\big(\iota(x), \iota(y)\big) \;=\; d_{\mathcal{M}_1}(x,y) \qquad \text{for all } x,y\in\mathcal{M}_1.
\]
Therefore the tree length functional is preserved:
\[
L_{\mathcal{M}_2}(\iota(V)) \;=\; L_{\mathcal{M}_1}(V).
\]
In particular, $\iota$ maps minimizing configurations in $\mathcal{M}_1$ bijectively to minimizing configurations in $\mathcal{M}_2$.  
If the minimizer is unique in $\mathcal{M}_1$, its image is unique in $\mathcal{M}_2$.
\end{proof}

Since we have established in Theorem~\ref{thm:uniqueness} that the minimizing configuration of Steiner points is unique in the Klein disk model, and since all standard models of hyperbolic space are mutually isometric, Lemma~\ref{lem:model_independence} implies the following general statement:\\

\vspace{-6pt}
\begin{center}
$\boxed{\parbox{0.8\textwidth}{\centering \emph{In hyperbolic space, given a fixed terminal set and a fixed tree topology, the optimal locations of the Steiner points are unique, independent of the chosen model.}}}$
\end{center}

\section{Additional Experimental Details}\label{app:exp_det}

This appendix provides implementation details for the three methods compared in our experiments: Randomized HyperSteiner (RHS), the original HyperSteiner (HS), and Neighbor Joining (NJ). 

\subsection{Hyperbolic SMT Heuristic Methods}

\paragraph{Randomized HyperSteiner (RHS).} 
We configure the global optimization procedure presented in Section~\ref{sec:glob_ref}, with the following hyperparameters, selected through tuning: a maximum of 10{,}000 epochs, learning rate $\eta = 1 \times 10^{-2}$, early stopping with patience of 100 epochs, and a convergence threshold of $\delta = 1 \times 10^{-6}$.

Although optimal learning rates are dataset-dependent, extensive hyperparameter tuning shows that $\eta = 1\times 10^{-2}$ performs consistently well across all datasets. As noted in Appendix~\ref{app:reim_desc}, more advanced optimization schemes (e.g., Armijo backtracking \citep{absil2008optimization} or Adam via trivializations \citep{lezcano2019trivializations}) could further reduce this dependence. 

For constructing full Steiner trees on four terminal points, we adopt the recursive decomposition approach of \citet{halverson2005steiner} with three recursion levels as recommended by \citet{garcia2025hypersteiner}. Moreover, for the three-point FSTs, we will solve the system of isoptic curves \eqref{eq:system}, using Powell's hybrid method \citep{Powell1970NumericalMethodsNonlinearAlgebraicEquations}. While \citet{garcia2025hypersteiner} demonstrated that Polynomial Homotopy Continuation methods \citep{chen2015homotopy} can achieve better solutions, the computational overhead outweighs the minor performance gains for our application.

\paragraph{Original HyperSteiner (HS).}
For consistency, we configure HS identically to RHS: both use the recursive solver for four-point trees and Powell’s method for three-point isoptic curves. 

\subsection{Neighbor Joining Method}
\label{app:nj}

\paragraph{Algorithm Overview.}
Neighbor Joining (NJ) \citep{Saitou1987neighborjoining} is a greedy agglomerative algorithm for phylogenetic tree reconstruction from pairwise distances. The method maintains a distance matrix and iteratively merges pairs of nodes into internal nodes using the Q-criterion, which corrects for varying evolutionary rates. NJ runs in $\mathcal{O}(n^3)$ time and produces a unique topology without requiring a molecular clock assumption.

\paragraph{Implementation and Adaptation.} NJ produces a tree topology together with edge lengths, yet the tree itself is not embedded in hyperbolic space. In line with \cite{Mimori2023Geophy}, we therefore disregard the edge lengths and retain only the topology. We subsequently extend NJ by embedding this topology into hyperbolic space via Riemannian gradient descent as follows:
\begin{enumerate}
    \item \textbf{Distance Matrix Construction.} Compute pairwise distances between all terminal points using the hyperbolic metric in the Klein disk.
    \item \textbf{Topology Inference.} Apply standard NJ to the distance matrix, yielding a tree topology $T_{NJ}$ with internal nodes (interpreted as Steiner points). We implement NJ via the Biotite Python library \citep{kunzmann2018biotite}.
    \item \textbf{Geometric Refinement Embedding.} Fixing topology $T_{NJ}$, optimize Steiner point positions to minimize total tree length in hyperbolic space. We initialize positions by sampling from a pseudo-hyperbolic Gaussian $\mathcal{G}(0,0.1)$ and optimize via Riemannian gradient descent with learning rate $\eta = 1$ for 10{,}000 epochs---see Section~\ref{sec:glob_ref} for further details regarding the optimization problem given a fixed topology.
\end{enumerate}

\newpage 
\section{Additional Results}
\label{app:exp}
\subsection{Performance Across Data Distributions}

\begin{wrapfigure}{r}{0.3\linewidth}
  \centering
  \vspace{-50pt}
  \includegraphics[width=0.75\linewidth]{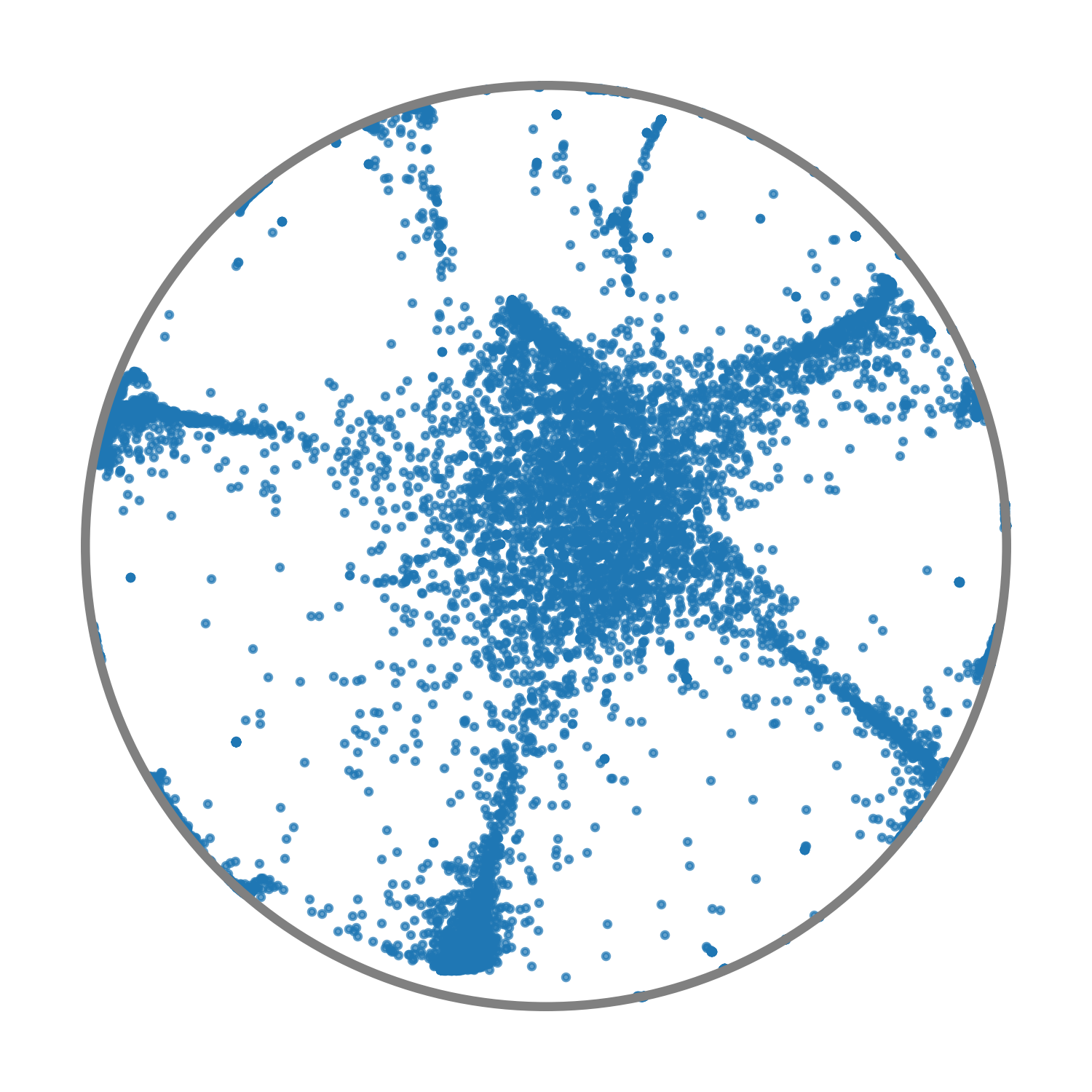}
  \caption{Planaria dataset.}
  \label{fig:planaria_plot}
  \vspace{-30pt}
\end{wrapfigure}

To complement the analysis in Section~\ref{sec:comp}, we provide additional visualizations and detailed numerical results. 

Figure \ref{fig:planaria_plot} shows the embedding of the \emph{Planaria} single-cell RNA-sequencing data \citep{Plass2018_planaria_SingleCellData}, while Tables \ref{tab:exp1} and \ref{tab:exp2} present the performance metrics for the synthetic experiments, including the percentage of reduction in tree-length over the Minimum Spanning Tree (RED) and computational efficiency measured by CPU wall time.

\vspace{30pt}
\begin{table*}[hb!]
\centering
\resizebox{\textwidth}{!}{
\begin{tabular}{lcccccc cccc}
\toprule
\multicolumn{1}{c}{} & \multicolumn{2}{c}{NJ} & \multicolumn{2}{c}{HS} & \multicolumn{2}{c}{RHS (ours)} & \multicolumn{2}{c}{vs HS (RED)} & \multicolumn{2}{c}{vs NJ (RED)} \\
\cmidrule(lr){2-3} \cmidrule(lr){4-5} \cmidrule(lr){6-7} \cmidrule(lr){8-9} \cmidrule(lr){10-11}
$|P|$ & RED\,$(\uparrow)$ & CPU & RED\,$(\uparrow)$ & CPU & RED\,$(\uparrow)$ & CPU & $p_{\text{Welch}}$ & $p_{\text{TOST}}$ & $p_{\text{Welch}}$ & $p_{\text{TOST}}$ \\
\midrule
    10 & $1.46_{\pm 1.12}$ & $5.50_{\pm 0.03}$ & $2.72_{\pm 0.97}$ & $0.03_{\pm 0.01}$ & $3.17_{\pm 1.59}$ & $4.60_{\pm 3.40}$ & $.457$ & $.183$ & $\mathbf{.013}$ & $.867$ \\
    20 & $-2.60_{\pm 3.00}$ & $5.68_{\pm 0.04}$ & $1.83_{\pm 0.70}$ & $0.06_{\pm 0.01}$ & $3.15_{\pm 0.56}$ & $9.87_{\pm 7.00}$ & $\mathbf{<\!.001}$ & $.863$ & $\mathbf{<\!.001}$ & $1.000$ \\
    30 & $-1.70_{\pm 2.31}$ & $5.86_{\pm 0.03}$ & $2.12_{\pm 0.60}$ & $0.10_{\pm 0.02}$ & $3.03_{\pm 1.09}$ & $12.65_{\pm 6.94}$ & $\mathbf{.036}$ & $.411$ & $\mathbf{<\!.001}$ & $1.000$ \\
    40 & $-2.62_{\pm 1.35}$ & $5.95_{\pm 0.04}$ & $2.17_{\pm 0.71}$ & $0.13_{\pm 0.02}$ & $3.30_{\pm 0.75}$ & $17.59_{\pm 6.97}$ & $\mathbf{.003}$ & $.652$ & $\mathbf{<\!.001}$ & $1.000$ \\
    50 & $-4.37_{\pm 1.87}$ & $6.14_{\pm 0.04}$ & $2.26_{\pm 0.43}$ & $0.17_{\pm 0.02}$ & $3.08_{\pm 0.34}$ & $26.21_{\pm 8.67}$ & $\mathbf{<\!.001}$ & $.157$ & $\mathbf{<\!.001}$ & $1.000$ \\
    60 & $-3.38_{\pm 1.09}$ & $6.29_{\pm 0.03}$ & $2.56_{\pm 0.57}$ & $0.20_{\pm 0.02}$ & $2.72_{\pm 0.48}$ & $26.95_{\pm 8.40}$ & $.506$ & $\mathbf{.001}$$^{\approx}$ & $\mathbf{<\!.001}$ & $1.000$ \\
    70 & $-3.28_{\pm 1.48}$ & $6.51_{\pm 0.04}$ & $2.45_{\pm 0.54}$ & $0.22_{\pm 0.01}$ & $2.68_{\pm 0.66}$ & $31.09_{\pm 11.84}$ & $.405$ & $\mathbf{.005}$$^{\approx}$ & $\mathbf{<\!.001}$ & $1.000$ \\
    80 & $-3.89_{\pm 1.79}$ & $6.77_{\pm 0.03}$ & $2.43_{\pm 0.53}$ & $0.26_{\pm 0.02}$ & $3.07_{\pm 0.39}$ & $38.55_{\pm 8.45}$ & $\mathbf{.007}$ & $.051$ & $\mathbf{<\!.001}$ & $1.000$ \\
    90 & $-3.46_{\pm 1.65}$ & $7.01_{\pm 0.04}$ & $2.72_{\pm 0.50}$ & $0.30_{\pm 0.02}$ & $2.87_{\pm 0.57}$ & $66.07_{\pm 30.15}$ & $.540$ & $\mathbf{.001}$$^{\approx}$ & $\mathbf{<\!.001}$ & $1.000$ \\
    100 & $-3.77_{\pm 1.75}$ & $7.24_{\pm 0.03}$ & $2.45_{\pm 0.35}$ & $0.33_{\pm 0.02}$ & $2.84_{\pm 0.47}$ & $70.31_{\pm 14.92}$ & $.051$ & $\mathbf{.002}$$^{\approx}$ & $\mathbf{<\!.001}$ & $1.000$ \\
    150 & $-4.17_{\pm 1.10}$ & $8.46_{\pm 0.03}$ & $2.27_{\pm 0.20}$ & $0.48_{\pm 0.02}$ & $2.48_{\pm 0.22}$ & $139.87_{\pm 49.03}$ & $\mathbf{.039}$ & $\mathbf{<\!.001}$$^{\approx}$ & $\mathbf{<\!.001}$ & $1.000$ \\
    200 & $-5.01_{\pm 1.40}$ & $11.09_{\pm 0.06}$ & $2.43_{\pm 0.26}$ & $0.66_{\pm 0.04}$ & $2.64_{\pm 0.23}$ & $341.42_{\pm 96.02}$ & $.072$ & $\mathbf{<\!.001}$$^{\approx}$ & $\mathbf{<\!.001}$ & $1.000$ \\
\bottomrule
\end{tabular}
}
\caption{Sampling $|P|$ points from a centered hyperbolic Gaussian $\mathcal{G}(0, 0.5)$. RED: reduction over MST (\%). CPU: total CPU time (sec.). Significance columns report Welch $t$-test and TOST p-values (RHS vs.\ the indicated baseline, RED metric, $n{=}10$, $\varepsilon{=}1.0$). \textbf{Bold} indicates $p{<}0.05$; ${}^{\approx}$ marks TOST equivalence (90\,\% CI $\subseteq (-\varepsilon,+\varepsilon)$ and $p_{\text{TOST}}{<}0.05$).}
\label{tab:exp1}
\end{table*}

\begin{table*}[b!]
\centering
\resizebox{\textwidth}{!}{
\begin{tabular}{lcccccc cccc}
\toprule
\multicolumn{1}{c}{} & \multicolumn{2}{c}{NJ} & \multicolumn{2}{c}{HS} & \multicolumn{2}{c}{RHS (ours)} & \multicolumn{2}{c}{vs HS (RED)} & \multicolumn{2}{c}{vs NJ (RED)} \\
\cmidrule(lr){2-3} \cmidrule(lr){4-5} \cmidrule(lr){6-7} \cmidrule(lr){8-9} \cmidrule(lr){10-11}
$|P|$ & RED\,$(\uparrow)$ & CPU & RED\,$(\uparrow)$ & CPU & RED\,$(\uparrow)$ & CPU & $p_{\text{Welch}}$ & $p_{\text{TOST}}$ & $p_{\text{Welch}}$ & $p_{\text{TOST}}$ \\
\midrule
    10 & $40.49_{\pm 0.03}$ & $5.52_{\pm 0.04}$ & $17.75_{\pm 4.57}$ & $0.02_{\pm 0.01}$ & $28.33_{\pm 18.54}$ & $25.70_{\pm 13.30}$ & $.110$ & $.928$ & $.068$ & $.955$ \\
    20 & $40.10_{\pm 0.06}$ & $5.73_{\pm 0.05}$ & $14.35_{\pm 4.61}$ & $0.05_{\pm 0.02}$ & $40.01_{\pm 0.09}$ & $61.94_{\pm 16.77}$ & $\mathbf{<\!.001}$ & $1.000$ & $\mathbf{.018}$ & $\mathbf{<\!.001}$$^{\approx}$ \\
    30 & $39.93_{\pm 0.08}$ & $5.91_{\pm 0.04}$ & $14.74_{\pm 5.75}$ & $0.08_{\pm 0.06}$ & $39.84_{\pm 0.13}$ & $95.05_{\pm 27.50}$ & $\mathbf{<\!.001}$ & $1.000$ & $.082$ & $\mathbf{<\!.001}$$^{\approx}$ \\
    40 & $39.70_{\pm 0.09}$ & $6.06_{\pm 0.03}$ & $16.63_{\pm 4.46}$ & $0.09_{\pm 0.03}$ & $39.71_{\pm 0.12}$ & $119.86_{\pm 53.15}$ & $\mathbf{<\!.001}$ & $1.000$ & $.836$ & $\mathbf{<\!.001}$$^{\approx}$ \\
    50 & $39.53_{\pm 0.06}$ & $6.00_{\pm 0.02}$ & $17.88_{\pm 5.29}$ & $0.10_{\pm 0.04}$ & $39.52_{\pm 0.05}$ & $116.38_{\pm 45.92}$ & $\mathbf{<\!.001}$ & $1.000$ & $.690$ & $\mathbf{<\!.001}$$^{\approx}$ \\
    60 & $39.22_{\pm 0.10}$ & $6.37_{\pm 0.04}$ & $14.95_{\pm 4.60}$ & $0.11_{\pm 0.02}$ & $39.36_{\pm 0.11}$ & $87.99_{\pm 9.33}$ & $\mathbf{<\!.001}$ & $1.000$ & $\mathbf{.008}$ & $\mathbf{<\!.001}$$^{\approx}$ \\
    70 & $39.01_{\pm 0.16}$ & $6.55_{\pm 0.05}$ & $14.42_{\pm 4.27}$ & $0.11_{\pm 0.03}$ & $39.26_{\pm 0.05}$ & $118.26_{\pm 24.23}$ & $\mathbf{<\!.001}$ & $1.000$ & $\mathbf{<\!.001}$ & $\mathbf{<\!.001}$$^{\approx}$ \\
    80 & $38.55_{\pm 0.11}$ & $6.77_{\pm 0.04}$ & $14.47_{\pm 4.29}$ & $0.12_{\pm 0.03}$ & $39.11_{\pm 0.14}$ & $111.56_{\pm 12.72}$ & $\mathbf{<\!.001}$ & $1.000$ & $\mathbf{<\!.001}$ & $\mathbf{<\!.001}$$^{\approx}$ \\
    90 & $38.18_{\pm 0.18}$ & $6.96_{\pm 0.06}$ & $15.85_{\pm 3.70}$ & $0.14_{\pm 0.03}$ & $39.03_{\pm 0.08}$ & $171.84_{\pm 59.68}$ & $\mathbf{<\!.001}$ & $1.000$ & $\mathbf{<\!.001}$ & $\mathbf{.016}$$^{\approx}$ \\
    100 & $37.50_{\pm 0.28}$ & $7.26_{\pm 0.07}$ & $13.91_{\pm 4.99}$ & $0.13_{\pm 0.04}$ & $38.94_{\pm 0.12}$ & $183.57_{\pm 60.30}$ & $\mathbf{<\!.001}$ & $1.000$ & $\mathbf{<\!.001}$ & $1.000$ \\
    150 & $24.02_{\pm 1.24}$ & $8.48_{\pm 0.05}$ & $13.22_{\pm 6.21}$ & $0.16_{\pm 0.02}$ & $38.47_{\pm 0.10}$ & $205.34_{\pm 52.63}$ & $\mathbf{<\!.001}$ & $1.000$ & $\mathbf{<\!.001}$ & $1.000$ \\
    200 & $-8.08_{\pm 1.09}$ & $10.51_{\pm 0.06}$ & $13.62_{\pm 4.43}$ & $0.18_{\pm 0.03}$ & $38.17_{\pm 0.17}$ & $265.46_{\pm 80.59}$ & $\mathbf{<\!.001}$ & $1.000$ & $\mathbf{<\!.001}$ & $1.000$ \\
\bottomrule
\end{tabular}
}
\caption{Sampling $|P|$ points from a mixture of $10$ hyperbolic Gaussians $\mathcal{G}(\mu_{10,k}(1-10^{-10}), 0.1)$, $k \in \{1, \ldots, 10\}$. RED: reduction over MST (\%). CPU: total CPU time (sec.). Significance columns report Welch $t$-test and TOST p-values (RHS vs.\ the indicated baseline, RED metric, $n{=}10$, $\varepsilon{=}1.0$). \textbf{Bold} indicates $p{<}0.05$; ${}^{\approx}$ marks TOST equivalence (90\,\% CI $\subseteq (-\varepsilon,+\varepsilon)$ and $p_{\text{TOST}}{<}0.05$). }
\label{tab:exp2}
\end{table*}

\newpage

For each configuration, we assess statistical differences between RHS and each baseline via two complementary tests: the Welch--Satterthwaite two-sample
$t$-test \citep{welch1938, satterthwaite1946}, which detects significant
differences in means without assuming equal variances across groups, and the
Two One-Sided Tests (TOST) procedure \citep{schuirmann1987}, which declares
\emph{practical equivalence} when the 90\% confidence interval of the mean
difference lies within $[-\varepsilon, +\varepsilon]$ (in our case we choose $\varepsilon{=}1.0$).
A significant Welch $p$-value with a non-significant TOST indicates a
detectable difference; a significant TOST (${}^{\approx}$) with a
non-significant Welch confirms practical equivalence; both significant implies
a real but negligible gap; neither significant is inconclusive.

\subsection{Near Boundary Performance Analysis}
\label{app:near_res}

We present in Figure \ref{fig:exp3_data} some examples of the sampled data considered for the first experiment (\emph{Approaching the Theoretical Limit}) of Section \ref{sec:boundary}, and in Table \ref{tab:exp3_results} the corresponding results. Similarly, for the second experiment (\emph{Characterizing the Transition Zone}), we illustrate in Figure \ref{fig:expTransition_data} some datasets employed, and in Figure \ref{fig:ConvergenceAnalysis_indv} the associated results.

The standard deviations reported in Table~\ref{tab:exp3_results} reveal that RHS not only achieves higher reductions but also maintains low variance across random seeds; the exceptions are failures arising from numerical instability at the boundary, in which case the method does not return a valid solution.  HS shows a different but related fragility: because it relies directly on Delaunay triangulations, it also becomes unstable near the boundary, but instead of terminating, it typically produces bad tree candidates. Stochastic expansion therefore contributes not only to solution quality but also to robustness, even though both heuristics remain sensitive to numerical issues in extreme boundary regimes. In contrast, NJ is less prompt to numerical instability in these regular polygon configurations near the boundary. Further discussion regarding the numerical instabilities at extreme boundary configurations can be found in Appendix~\ref{app:near_dis}.

However, as highlighted in Figure \ref{fig:ConvergenceAnalysis_indv}, NJ performances degrade significatively when multiple clusters are introduced. For sufficiently high $t$ (around 0.95 for RHS, 0.99 for HS, and already from 0.6 for NJ), and for a fixed number of clusters between 3 and 6, all three methods show improved RED as $t$ increases. This trend also holds for RHS and HS when the number of clusters is higher. In contrast, NJ deteriorates sharply at the extreme boundary ($t=1-10^{-10}$) when the number of clusters is more than 6. Moreover, for fixed high $t > 1-10^{-5}$, only RHS continues to benefit from an increasing number of clusters, while HS and NJ plateau or decline.
In summary, RHS is the only method that consistently benefits both from being close to the boundary and from increasing the number of clusters. HS improves with boundary proximity, but saturates as the number of clusters grows. NJ initially improves when points are moderately close to the boundary, but its performance collapses in the extreme limit, showing that it cannot handle high curvature combined with clustering.

\begin{figure*}[h]
\centering
  \begin{subfigure}[t]{0.3\linewidth}
     \centering
     \includegraphics[width=\linewidth]{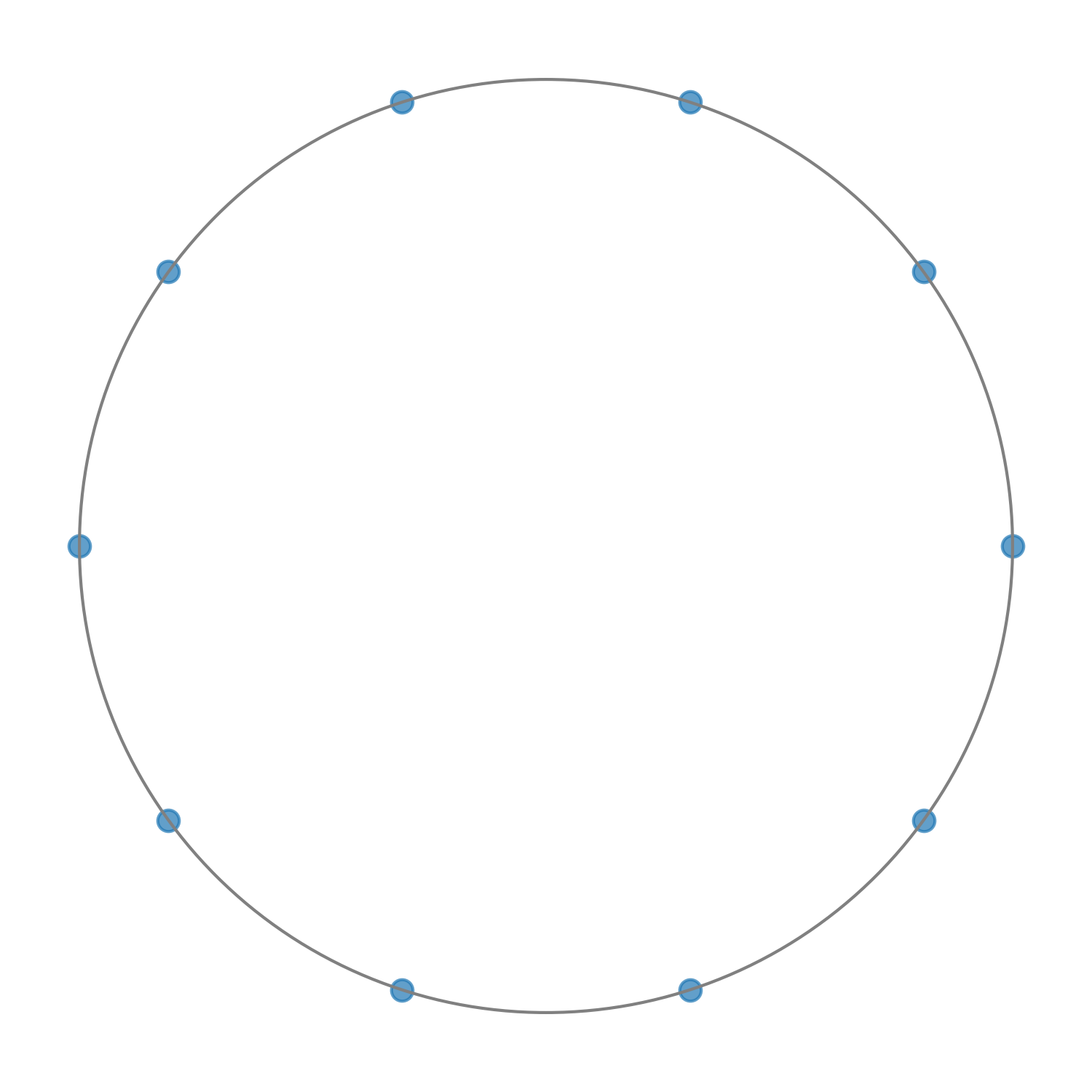}
     \caption*{$|P|=10$}
     \label{fig:exp3_10}
 \end{subfigure}
 \hfill
 \begin{subfigure}[t]{0.3\linewidth}
     \centering
     \includegraphics[width=\linewidth]{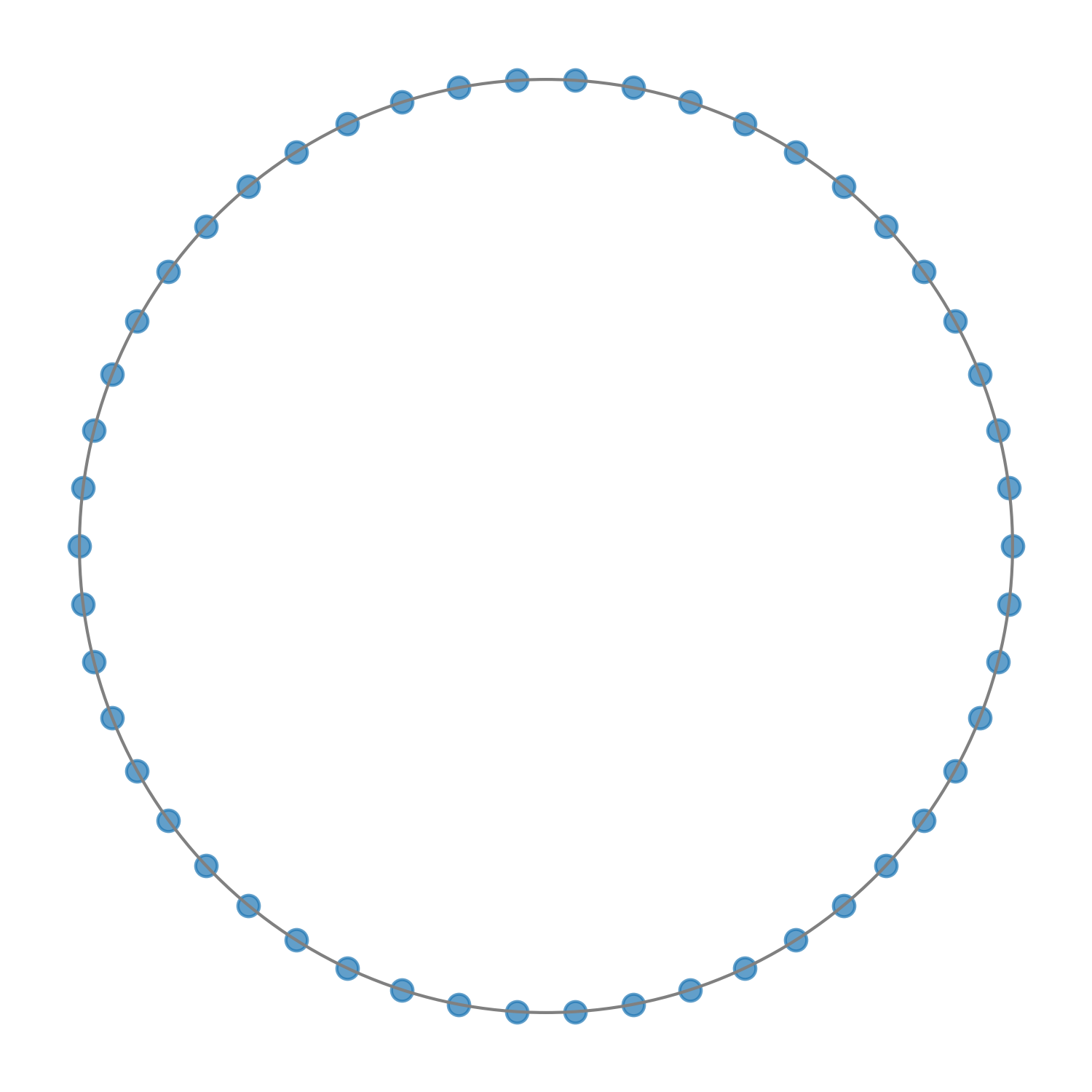}
     \caption*{$|P|=50$}
     \label{fig:exp3_50}
 \end{subfigure}
\hfill
 \begin{subfigure}[t]{0.3\linewidth}
     \centering
     \includegraphics[width=\linewidth]{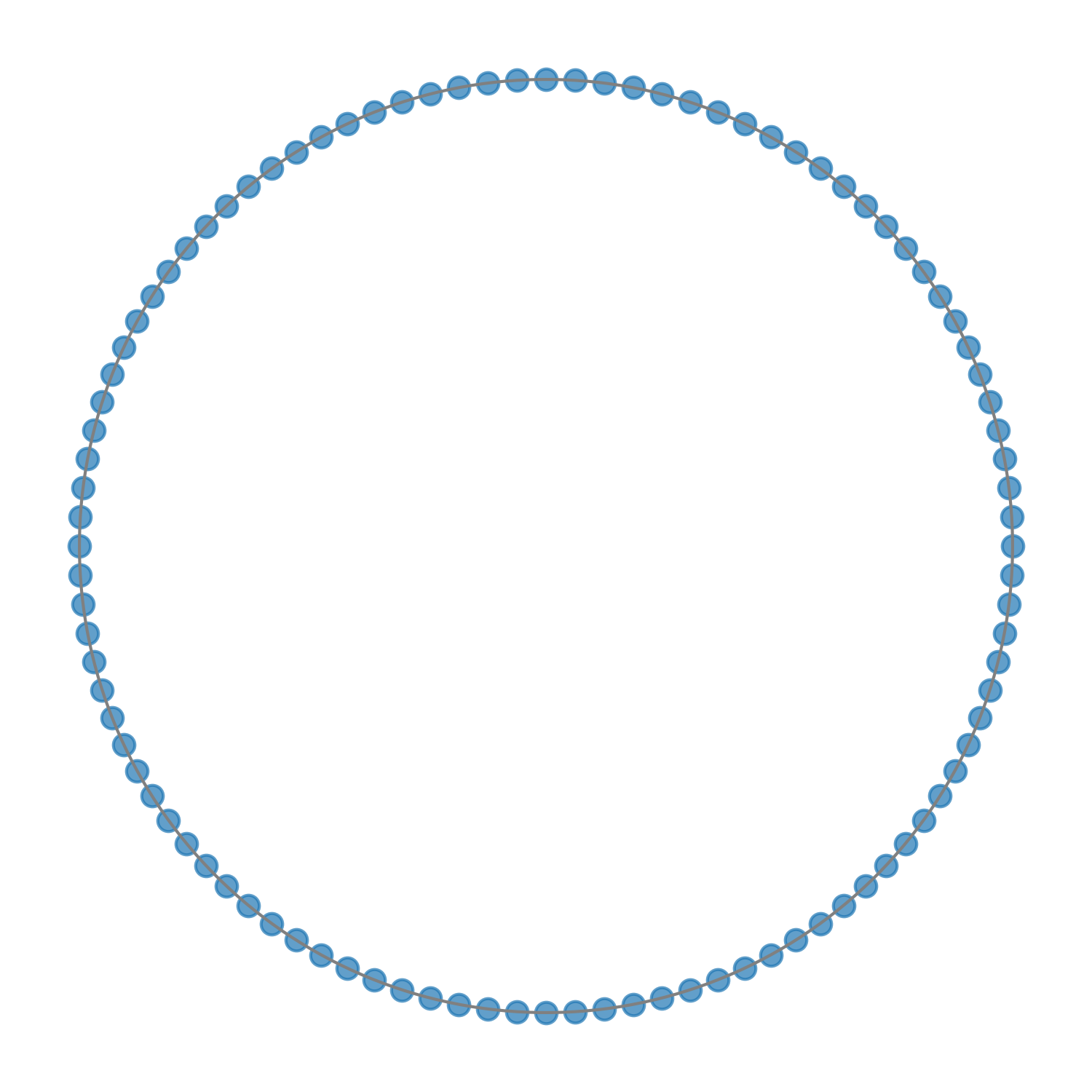}
     \caption*{$|P|=100$}
     \label{fig:exp3_100}
 \end{subfigure}

\caption{Visualization of some of the datasets considered in the \emph{Approaching the Theoretical Limit} experiment. Sampling $|P|\in \{10, 50, 100\}$ points from a mixture of hyperbolic Gaussians $\mathcal{G}(\mu_{|P|,k}(1-10^{-10}), 0.1)$ with $k \in \{1, \cdots, |P|\}$.}
\label{fig:exp3_data}
\end{figure*}

\begin{table*}[ht!]
\centering
\resizebox{0.9\textwidth}{!}{
\begin{tabular}{lcccccc}
\toprule
\multicolumn{1}{c}{} & \multicolumn{2}{c}{NJ} & \multicolumn{2}{c}{HS} & \multicolumn{2}{c}{RHS (ours)} \\
\cmidrule(lr){2-3}
\cmidrule(lr){4-5}
\cmidrule(lr){6-7}
$|P|$ & RED $(\uparrow)$ & CPU & RED $(\uparrow)$ & CPU & RED $(\uparrow)$ & CPU \\
\midrule
4  & $31.31_{\pm 0.07}$ & $5.54_{\pm 0.02}$ & $31.34_{\pm 0.08}$ & $0.01_{\pm 0.002}$ & $31.39_{\pm 0.11}$ & $1.59_{\pm 0.51}$ \\
5  & $34.86_{\pm 0.09}$ & $5.58_{\pm 0.02}$ & $23.26_{\pm 0.02}$ & $0.01_{\pm 0.003}$ & $34.94_{\pm 0.08}$ & $10.26_{\pm 1.44}$ \\
6  & $36.95_{\pm 0.10}$ & $5.60_{\pm 0.07}$ & $18.28_{\pm 0.12}$ & $0.01_{\pm 0.01}$ & $37.01_{\pm 0.10}$ & $15.46_{\pm 1.64}$ \\
7  & $38.30_{\pm 0.07}$ & $5.59_{\pm 0.02}$ & $17.61_{\pm 3.49}$ & $0.01_{\pm 0.005}$ & $38.33_{\pm 0.01}$ & $13.09_{\pm 4.26}$ \\
8  & $39.26_{\pm 0.03}$ & $5.56_{\pm 0.00}$ & $10.81_{\pm 6.10}$ & $0.01_{\pm 0.004}$ & $39.28_{\pm 0.02}$ & $24.27_{\pm 8.48}$ \\
9  & $39.95_{\pm 0.01}$ & $5.59_{\pm 0.03}$ & $13.05_{\pm 2.58}$ & $0.01_{\pm 0.01}$ & $39.96_{\pm 0.02}$ & $52.79_{\pm 18.82}$ \\
10 & $40.48_{\pm 0.01}$ & $5.59_{\pm 0.01}$ & $11.60_{\pm 2.38}$ & $0.02_{\pm 0.005}$ & $40.51_{\pm 0.02}$ & $35.03_{\pm 10.10}$ \\
20 & $42.48_{\pm 0.01}$ & $5.76_{\pm 0.04}$ & $20.86_{\pm 1.99}$ & $0.17_{\pm 0.03}$ & $42.59_{\pm 0.04}$ & $92.29_{\pm 14.17}$ \\
30 & $42.93_{\pm 0.004}$ & $5.84_{\pm 0.03}$ & $21.04_{\pm 2.51}$ & $0.31_{\pm 0.05}$ & $43.03_{\pm 0.11}$ & $89.08_{\pm 25.65}$ \\
40 & $43.10_{\pm 0.02}$ & $6.01_{\pm 0.03}$ & $17.56_{\pm 2.21}$ & $0.33_{\pm 0.02}$ & $43.29_{\pm 0.08}$ & $159.69_{\pm 41.75}$ \\
50 & $43.16_{\pm 0.01}$ & $6.21_{\pm 0.03}$ & $20.07_{\pm 1.17}$ & $0.50_{\pm 0.02}$ & $43.26_{\pm 0.06}$ & $217.99_{\pm 52.70}$ \\
60 & $43.14_{\pm 0.01}$ & $6.40_{\pm 0.03}$ & $18.80_{\pm 1.16}$ & $0.58_{\pm 0.03}$ & $43.26_{\pm 0.11}$ & $214.50_{\pm 64.05}$ \\
70 & $43.12_{\pm 0.003}$ & $6.63_{\pm 0.02}$ & $18.84_{\pm 1.01}$ & $0.63_{\pm 0.06}$ & $43.39_{\pm 0.04}$ & $287.02_{\pm 29.79}$ \\
80 & $43.09_{\pm 0.01}$ & $6.81_{\pm 0.04}$ & $17.09_{\pm 1.30}$ & $0.75_{\pm 0.06}$ & $43.09_{\pm 0.21}$ & $239.36_{\pm 76.02}$ \\
90 & $43.05_{\pm 0.01}$ & $7.03_{\pm 0.01}$ & $18.65_{\pm 0.75}$ & $0.83_{\pm 0.04}$ & $43.15_{\pm 0.13}$ & $304.83_{\pm 106.17}$ \\
100 & $43.02_{\pm 0.009}$ & $7.20_{\pm 0.03}$ & $18.00_{\pm 0.77}$ & $0.80_{\pm 0.01}$ & $42.99_{\pm 0.11}$ & $310.98_{\pm 54.61}$ \\
\bottomrule
\end{tabular}
}
\caption{Sampling $|P|$ points from a mixture of hyperbolic Gaussians $\mathcal{G}(\mu_{|P|, k}(1-10^{-10}), 0.1)$, $k \in \{1, \cdots, |P|\}$ near the boundary of the hyperbolic disk with only one point sampled per Gaussian. RED: reduction over MST (\%). CPU: total CPU time (sec.).}
\label{tab:exp3_results}
\end{table*}


\begin{figure}[b!]
    \centering    \includegraphics[width=0.8\linewidth]{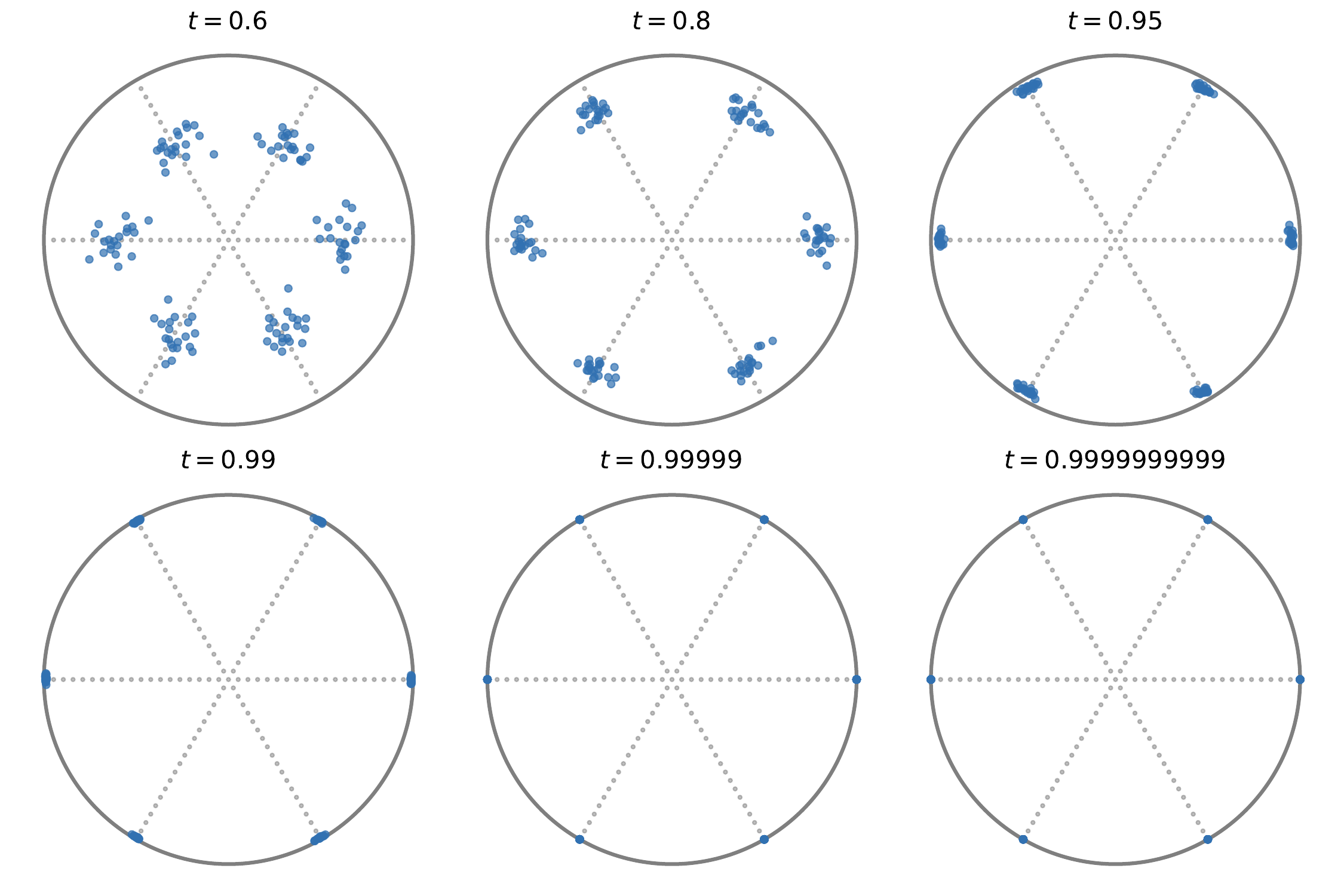}
    \caption{Visualization of the data sets considered in the \emph{Characterizing the Transition Zone} experiment for the $d=6$ case. Sampling $20$ points from each hyperbolic Gaussian $\mathcal{G}(\mu_{6,k}(t), 0.1)$ with $k \in \{1, \cdots, 6\}$, and $t\in\{0.6, 0.8, 0.95, 0.99, 1-10^{-5}, 1-10^{-10}\}$. }
    \label{fig:expTransition_data}
\end{figure}

\begin{figure*}[h!]
\centering

\begin{subfigure}[t]{0.49\linewidth}
    \centering
    \includegraphics[width=\linewidth]{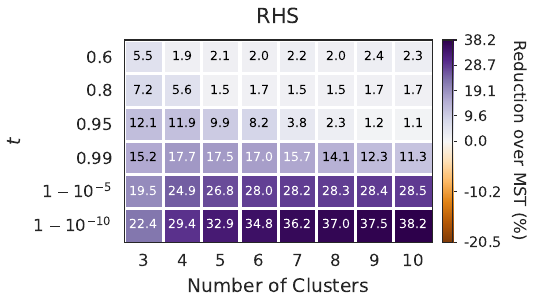}
\end{subfigure}%
\begin{subfigure}[t]{0.49\linewidth}
    \centering
    \includegraphics[width=\linewidth]{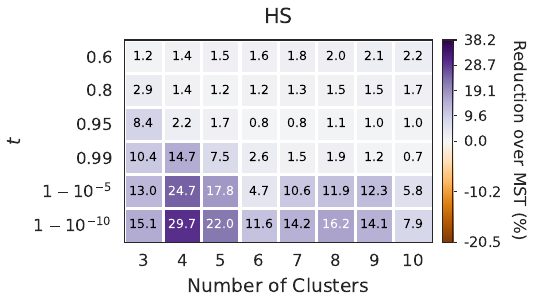}
\end{subfigure}%

\begin{subfigure}[t]{0.49\linewidth}
    \centering
    \includegraphics[width=\linewidth]{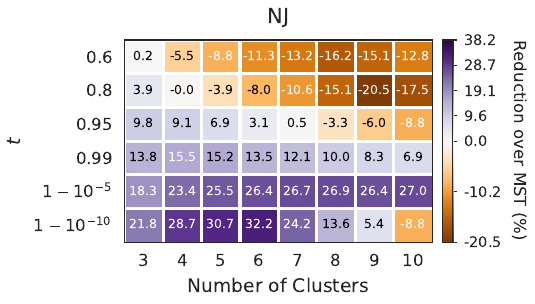}
\end{subfigure}%

\caption{
Convergence analysis for a mixture of $\mathcal{G}(\mu_{d, k}(t), 0.1)$, $k \in \{1, \cdots, d\}$ for $d \in \{3, \cdots, 10\}$ and with an increasing $t$, sampling 20 points per Gaussian.  Values correspond to reduction over MST (\%).}
\label{fig:ConvergenceAnalysis_indv}
\end{figure*}

\subsection{Discussion on Near Boundary Numerical Stability}\label{app:near_dis}
As reported in Appendix~\ref{app:near_res}, numerical instabilities arise only in extreme near-boundary regimes, for example when $t = 1 - 10^{-10}$. This setting is included solely as a theoretical stress test for the upper-bound ratio and is not relevant in practical applications. For less extreme values, such as $t \geq 1 - 10^{-5}$, we do not observe instability.

Indeed, as also noted by \citet{pmlr-v202-mishne23a}, such numerical effects depend in part on the choice of hyperbolic model. We adopt the Klein--Beltrami model because it simplifies the computation of Hyperbolic Delaunay Triangulations and isoptic curves \citep[Appendix~A.2]{garcia2025hypersteiner}. When additional numerical robustness is required, the method could instead be implemented in the hyperboloid model. We leave this variation to future work.

\subsection{Discussion on Scalability and Computational Complexity}
\label{app:scala}

\paragraph{Runtime comparison on synthetic experiments.} As shown in Tables~\ref{tab:exp1}, \ref{tab:exp2}, and \ref{tab:exp3_results}, the proposed method consistently achieves superior performance compared to both baselines. While its computational cost exceeds that of NJ and HS, this is expected for non-deterministic heuristics designed to compute Steiner Minimal Trees. To contextualize our computational requirements, our method requires approximately 187 seconds for 100 points (averaged across all configurations). In comparison, other non-deterministic approaches report varied costs: among DT-based methods, \cite{Beasley1994DTHeuristicSMT} requires 488.4 seconds for 100 points and \cite{Laarhoven2011Randomized} requires 40.2 seconds for 100 points, while \cite{yang2006hybrid} report that their evolutionary algorithm requires over 5,000 seconds for 100 points. Our method thus falls within the lower range of computational costs for non-deterministic heuristics in this domain.

\paragraph{Estimating the computational complexity.} As stated in Appendix~\ref{app:og_hyp} and Appendix~\ref{app:nj}, HyperSteiner has a computational complexity of $\mathcal{O}(|P|\log|P|)$~\citep{garcia2025hypersteiner}, while Neighbor Joining has a complexity of $\mathcal{O}(|P|^3)$ in its standard form~\citep{Studier1988NJ}, and improved versions achieve $\mathcal{O}(|P|^2)$ or $\mathcal{O}(|P|^{3/2}\log|P|)$ (respectively \citep{Elias2009FastNJ, Price2009FastTree}). As mentioned above, due to its stochastic aspect, we cannot compute a theoretical complexity for Randomized HyperSteiner, similar to the situation with its Euclidean stochastic analogs.

Following the methodology of \cite{Beasley1994DTHeuristicSMT}, we compute instead an empirical complexity by modeling the CPU runtime as a function of the number of points $|P|$, using a parametric form fitted via least-squares on log-transformed data. We focus on the results from Table 1, where the hyperbolic Gaussian sampling is locally Euclidean, allowing for rough comparison with Beasley \& Goffinet's empirical finding of $\mathcal{O}(|P|^{2.19})$ for their Euclidean stochastic method. In our case, the average computational time of Randomized HyperSteiner is empirically estimated to be $\mathcal{O}(|P|^{2.1}\log|P|)$ or $\mathcal{O}(|P|^{2.28})$. This confirms that our method is computationally more expensive than the deterministic baselines, as expected for stochastic heuristics, while remaining in the same complexity range as Euclidean stochastic approaches.

\paragraph{Equation \ref{eq:system} computational complexity.} Equation \eqref{eq:system} is solved via MINPACK’s \emph{hybrd} and \emph{hybrj} algorithms (implemented with scipy.optimize.fsolve in the Python library) which are modifications of  Powell's hybrid method. The latter is a version of Newton's method with a trust-region approach, and its computational complexity depends on the accuracy requested and the number of iterations needed to converge, which itself depends on the quality of the initialization. For one iteration of these algorithms, since we have a system of two (nonlinear) equations, the number of arithmetic operations (according to MINPACK's documentation from Burton et al., 1980) is about $11.5\times 2^2 + 1.3 \times 2^{3} = 56.4$ in addition with the number of operations needed to evaluate the functions $\varphi_{x,y,\frac{2\pi}{3}}(s)$ and $\varphi_{y,z,\frac{2\pi}{3}}(s)$. Thus, the computational cost per iteration is constant time. A reasonable value for the maximal number of iterations (given also by the MINPACK documentation) is $100\times (2+1) = 300$. In summary, for a system of two equations in two dimensions, the computational complexity of solving Equation~ \eqref{eq:system} is, essentially, $\mathcal{O}(1)$.

\paragraph{Compute-performance trade-off.} These results present practitioners with a clear trade-off: RHS should be selected when maximizing solution quality is the priority, while the vanilla HyperSteiner remains suitable for applications requiring fast approximations. Notably, as demonstrated in Section~\ref{sec:boundary}, on configurations approaching the boundary, the proposed method achieves performance improvements of 20\% to 40\% over the MST, as well as significant gains over the other baselines. Under such configurations, these substantial improvements justify the additional computational cost.

\subsection{Ablation Study: Insertion probabilities}
\label{app:ins_prob}

As stated in Appendix~\ref{app:exp_phase}, for the insertion probability range $[l,u]$ of our \emph{Expansion phase} in Algorithm \ref{alg:rad_hyper}, we chose the interval $[0.3,0.6]$ following \citet{Laarhoven2011Randomized}. We provide in Tables \ref{tab:ablation_proba} and \ref{tab:ablation_proba_cpu} an ablation study on the probability insertion $p$ for a typical configuration in the hyperbolic space, that is when sampling $|P|$ points from a centered hyperbolic Gaussian. Taken together, these two tables suggest that choosing $p$ between 0.3 and 0.6 provides a good trade-off between maximizing RED and minimizing computational cost.

\vspace{20pt}
\begin{table*}[h!]
    \centering
    \scriptsize
    \setlength{\tabcolsep}{4pt}
    \resizebox{\textwidth}{!}{%
    \begin{tabular}{lccccccccccc}
        \toprule
        $|P|$ & $p=0$ & $p=0.1$ & $p=0.2$ & $p=0.3$ & $p=0.4$ & $p=0.5$ & $p=0.6$ & $p=0.7$ & $p=0.8$ & $p=0.9$ & $p=1$ \\
        \midrule
        10  & $3.03_{\pm 1.42}$ & $3.30_{\pm 1.32}$ & $2.02_{\pm 0.97}$ & $3.29_{\pm 1.72}$ & $\mathbf{3.76_{\pm 0.94}}$ & $3.53_{\pm 1.78}$ & $3.07_{\pm 1.16}$ & $3.45_{\pm 1.89}$ & $1.68_{\pm 1.02}$ & $\mathbf{3.95_{\pm 1.87}}$ & $3.26_{\pm 1.28}$ \\
        20  & $2.81_{\pm 1.14}$ & $2.35_{\pm 0.84}$ & $2.79_{\pm 0.68}$ & $\mathbf{3.41_{\pm 1.26}}$ & $3.18_{\pm 0.79}$ & $3.34_{\pm 0.86}$ & $\mathbf{3.55_{\pm 1.30}}$ & $2.53_{\pm 0.85}$ & $2.87_{\pm 0.45}$ & $2.88_{\pm 0.68}$ & $2.60_{\pm 0.64}$ \\
        30  & $2.60_{\pm 0.75}$ & $2.69_{\pm 0.80}$ & $2.63_{\pm 0.90}$ & $2.23_{\pm 0.52}$ & $2.66_{\pm 0.80}$ & $2.83_{\pm 0.58}$ & $\mathbf{2.89_{\pm 0.73}}$ & $\mathbf{2.94_{\pm 0.61}}$ & $2.84_{\pm 0.81}$ & $2.81_{\pm 0.84}$ & $2.65_{\pm 0.85}$ \\
        40  & $2.17_{\pm 0.56}$ & $2.72_{\pm 0.89}$ & $2.89_{\pm 0.74}$ & $\mathbf{2.96_{\pm 0.70}}$ & $2.50_{\pm 0.35}$ & $2.81_{\pm 0.72}$ & $\mathbf{2.98_{\pm 0.56}}$ & $2.93_{\pm 0.66}$ & $2.71_{\pm 0.50}$ & $2.69_{\pm 0.51}$ & $2.37_{\pm 0.67}$ \\
        50  & $2.24_{\pm 0.59}$ & $2.83_{\pm 0.65}$ & $2.80_{\pm 0.71}$ & $\mathbf{3.36_{\pm 0.46}}$ & $2.59_{\pm 0.39}$ & $3.10_{\pm 0.65}$ & $3.08_{\pm 0.65}$ & $2.74_{\pm 0.51}$ & $2.90_{\pm 0.84}$ & $\mathbf{3.29_{\pm 0.74}}$ & $2.54_{\pm 0.57}$ \\
        60  & $2.08_{\pm 0.53}$ & $\mathbf{3.02_{\pm 0.53}}$ & $2.94_{\pm 0.63}$ & $2.72_{\pm 0.46}$ & $2.70_{\pm 0.57}$ & $2.90_{\pm 0.56}$ & $2.77_{\pm 0.58}$ & $2.77_{\pm 0.26}$ & $\mathbf{3.00_{\pm 0.55}}$ & $2.88_{\pm 0.56}$ & $2.84_{\pm 0.55}$ \\
        70  & $2.20_{\pm 0.52}$ & $2.86_{\pm 0.62}$ & $2.43_{\pm 0.58}$ & $\mathbf{3.25_{\pm 0.60}}$ & $2.82_{\pm 0.67}$ & $2.97_{\pm 0.54}$ & $2.43_{\pm 0.64}$ & $2.78_{\pm 0.57}$ & $\mathbf{3.01_{\pm 0.57}}$ & $2.85_{\pm 0.57}$ & $2.57_{\pm 0.64}$ \\
        80  & $2.21_{\pm 0.59}$ & $2.98_{\pm 0.49}$ & $2.81_{\pm 0.69}$ & $2.88_{\pm 0.63}$ & $2.77_{\pm 0.69}$ & $\mathbf{3.00_{\pm 0.51}}$ & $2.81_{\pm 0.51}$ & $2.73_{\pm 0.69}$ & $\mathbf{3.17_{\pm 0.51}}$ & $2.97_{\pm 0.63}$ & $2.64_{\pm 0.62}$ \\
        90  & $2.27_{\pm 0.50}$ & $2.95_{\pm 0.51}$ & $2.88_{\pm 0.56}$ & $2.89_{\pm 0.52}$ & $2.81_{\pm 0.51}$ & $2.88_{\pm 0.58}$ & $\mathbf{3.18_{\pm 0.57}}$ & $2.88_{\pm 0.48}$ & $2.86_{\pm 0.53}$ & $\mathbf{2.95_{\pm 0.45}}$ & $2.91_{\pm 0.44}$ \\
        100 & $2.29_{\pm 0.45}$ & $2.71_{\pm 0.59}$ & $\mathbf{2.98_{\pm 0.46}}$ & $2.79_{\pm 0.41}$ & $2.82_{\pm 0.40}$ & $2.66_{\pm 0.47}$ & $2.65_{\pm 0.42}$ & $2.91_{\pm 0.47}$ & $2.82_{\pm 0.36}$ & $\mathbf{2.95_{\pm 0.42}}$ & $2.67_{\pm 0.52}$ \\
        \midrule
        Avg & $2.39_{\pm 0.70}$ & $2.84_{\pm 0.72}$ & $2.72_{\pm 0.69}$ & $2.98_{\pm 0.73}$ & $2.86_{\pm 0.61}$ & $3.00_{\pm 0.72}$ & $2.94_{\pm 0.71}$ & $2.87_{\pm 0.70}$ & $2.79_{\pm 0.61}$ & $3.02_{\pm 0.73}$ & $2.70_{\pm 0.68}$ \\
        \bottomrule
    \end{tabular}%
    }
    \caption{Randomized HyperSteiner RED results when sampling $|P|$ points from a Hyperbolic Centered Gaussian (mean $\pm$ standard deviation) for insertion probabilities $p$ between 0 and 1. The two best results for each row are in bold.}
    \label{tab:ablation_proba}
\end{table*}

\begin{table*}[h!]
    \centering
    \scriptsize
    \setlength{\tabcolsep}{4pt}
    \resizebox{\textwidth}{!}{%
    \begin{tabular}{lccccccccccc}
        \toprule
        $|P|$ & $p=0$ & $p=0.1$ & $p=0.2$ & $p=0.3$ & $p=0.4$ & $p=0.5$ & $p=0.6$ & $p=0.7$ & $p=0.8$ & $p=0.9$ & $p=1$ \\
        \midrule
        10  & $8.25_{\pm 4.89}$ & $8.46_{\pm 6.28}$ & $9.38_{\pm 6.38}$ & $6.51_{\pm 5.42}$ & $5.03_{\pm 2.93}$ & $8.07_{\pm 5.13}$ & $7.93_{\pm 5.67}$ & $5.55_{\pm 4.69}$ & $3.62_{\pm 2.36}$ & $6.21_{\pm 4.38}$ & $5.42_{\pm 4.68}$ \\
        20  & $21.22_{\pm 18.91}$ & $21.34_{\pm 17.15}$ & $12.56_{\pm 11.00}$ & $12.96_{\pm 5.43}$ & $12.08_{\pm 4.61}$ & $8.38_{\pm 3.88}$ & $12.43_{\pm 7.75}$ & $12.82_{\pm 8.40}$ & $7.37_{\pm 2.10}$ & $12.19_{\pm 8.40}$ & $25.15_{\pm 46.83}$ \\
        30  & $23.10_{\pm 21.33}$ & $15.49_{\pm 7.27}$ & $12.30_{\pm 6.21}$ & $11.74_{\pm 6.01}$ & $11.15_{\pm 6.21}$ & $12.85_{\pm 5.29}$ & $11.45_{\pm 6.04}$ & $14.40_{\pm 7.32}$ & $15.62_{\pm 4.68}$ & $17.78_{\pm 8.55}$ & $17.07_{\pm 7.11}$ \\
        40  & $20.89_{\pm 18.82}$ & $22.11_{\pm 7.50}$ & $15.97_{\pm 7.03}$ & $13.71_{\pm 5.23}$ & $16.16_{\pm 8.94}$ & $14.51_{\pm 6.36}$ & $18.70_{\pm 7.59}$ & $20.33_{\pm 4.27}$ & $25.16_{\pm 11.66}$ & $22.72_{\pm 7.91}$ & $24.55_{\pm 7.50}$ \\
        50  & $35.75_{\pm 32.54}$ & $20.79_{\pm 6.46}$ & $17.66_{\pm 8.74}$ & $16.15_{\pm 5.77}$ & $20.57_{\pm 10.73}$ & $17.53_{\pm 5.78}$ & $24.03_{\pm 11.60}$ & $35.58_{\pm 10.01}$ & $43.23_{\pm 11.14}$ & $69.53_{\pm 15.59}$ & $74.89_{\pm 28.67}$ \\
        60  & $28.92_{\pm 26.24}$ & $26.41_{\pm 8.73}$ & $17.74_{\pm 5.68}$ & $18.37_{\pm 7.40}$ & $21.36_{\pm 4.92}$ & $24.06_{\pm 9.25}$ & $28.27_{\pm 6.51}$ & $36.43_{\pm 13.48}$ & $67.70_{\pm 31.42}$ & $73.89_{\pm 22.90}$ & $109.33_{\pm 33.30}$ \\
        70  & $55.84_{\pm 71.78}$ & $28.26_{\pm 8.94}$ & $19.50_{\pm 7.85}$ & $20.67_{\pm 6.88}$ & $23.42_{\pm 5.67}$ & $36.56_{\pm 20.86}$ & $40.92_{\pm 9.56}$ & $63.02_{\pm 20.79}$ & $85.85_{\pm 23.53}$ & $134.40_{\pm 39.08}$ & $226.67_{\pm 82.49}$ \\
        80  & $60.66_{\pm 109.63}$ & $30.58_{\pm 12.82}$ & $21.00_{\pm 5.69}$ & $21.98_{\pm 7.64}$ & $30.11_{\pm 8.40}$ & $30.95_{\pm 9.33}$ & $48.54_{\pm 17.12}$ & $78.43_{\pm 18.81}$ & $123.10_{\pm 41.94}$ & $210.40_{\pm 67.90}$ & $293.06_{\pm 75.58}$ \\
        90  & $80.30_{\pm 120.06}$ & $32.09_{\pm 13.54}$ & $33.62_{\pm 12.13}$ & $33.49_{\pm 6.74}$ & $38.06_{\pm 11.01}$ & $44.73_{\pm 12.22}$ & $76.83_{\pm 18.01}$ & $135.59_{\pm 38.33}$ & $261.36_{\pm 70.28}$ & $429.71_{\pm 87.88}$ & $917.67_{\pm 349.84}$ \\
        100 & $29.24_{\pm 13.20}$ & $39.82_{\pm 8.82}$ & $30.29_{\pm 10.15}$ & $41.65_{\pm 17.64}$ & $43.38_{\pm 7.20}$ & $53.57_{\pm 10.77}$ & $121.44_{\pm 38.52}$ & $195.24_{\pm 63.01}$ & $326.40_{\pm 102.91}$ & $925.18_{\pm 211.66}$ & $1936.48_{\pm 524.08}$ \\
        \midrule
        Avg & $36.42_{\pm 43.74}$ & $24.54_{\pm 9.75}$ & $19.00_{\pm 8.09}$ & $19.72_{\pm 7.42}$ & $22.13_{\pm 7.06}$ & $25.12_{\pm 8.89}$ & $39.05_{\pm 12.84}$ & $59.74_{\pm 18.91}$ & $95.94_{\pm 30.20}$ & $190.20_{\pm 47.42}$ & $363.03_{\pm 116.01}$ \\
        \bottomrule
    \end{tabular}%
    }
    \caption{Randomized HyperSteiner CPU time (seconds) when sampling $|P|$ points from a Hyperbolic Centered Gaussian (mean $\pm$ standard deviation) for insertion probabilities $p$ between 0 and 1.}
    \label{tab:ablation_proba_cpu}
\end{table*}
\newpage
\subsection{Ablation Study: Expansion schedule}
\label{app:exp_sche}

As we show in Appendix~\ref{app:exp_phase}, there are several choices for how many times we expand the Delaunay triangulation at each step. In Figure~\ref{fig:exps}, we compare three approaches: ``linear" ($n \mapsto n$) \citep{Beasley1994DTHeuristicSMT}, ``constant" ($n \mapsto 1$), and ours, which we call ``sqrt" ($n \mapsto \lfloor 2\sqrt{n} - 1 \rfloor$). While we provide justifications for selecting the ``sqrt" schedule in Appendix~\ref{app:exp_phase}, we will perform an ablation study in this section to further support this choice.

Hence, we provide in Table~\ref{tab:ab_schedule} a performance comparison on a mixture of $d$ hyperbolic Gaussians, each one centered at a vertex of a regular $d$-gon, with radial parameter $t=0.9$, and a standard deviation of $0.01$. We sample 20 points per Gaussian, which leads to $d$ clusters. This dataset resembles the ones of Section~\ref{sec:boundary} for our \emph{Near Boundary Performance Analysis}, making it a suitable benchmark for our purposes.

As shown in Table~\ref{tab:ab_schedule}, the ``linear" schedule can be more expressive than the ``sqrt" one, but its computational cost increases significantly as the number of clusters grows. Therefore, in line with our justification in Appendix~\ref{app:exp_phase}, we adopt the ``sqrt" schedule in our experiments as it provides the highest performance with a reasonable computational cost. Nevertheless, in practice, the ``constant" schedule could be employed to improve computational efficiency without a significant performance drop.
\vspace{10pt}
\begin{table*}[ht!]
\centering
\resizebox{0.85\textwidth}{!}{
\begin{tabular}{lcccccc}
\toprule
\multicolumn{1}{c}{} & \multicolumn{2}{c}{sqrt} & \multicolumn{2}{c}{constant} & \multicolumn{2}{c}{linear} \\
\cmidrule(lr){2-3}
\cmidrule(lr){4-5}
\cmidrule(lr){6-7}
$d$ & RED $(\uparrow)$ & CPU & RED $(\uparrow)$ & CPU & RED $(\uparrow)$ & CPU \\
\midrule
4 & $\mathbf{16.75}_{\pm 0.05}$ & $14.80_{\pm 6.98}$  & $16.62_{\pm 0.01}$ & $5.91_{\pm 1.00}$   & $16.68_{\pm 0.06}$ & $23.05_{\pm 7.42}$   \\
5 & $\mathbf{13.68}_{\pm 0.07}$ & $36.45_{\pm 15.14}$ & $13.61_{\pm 0.12}$ & $15.41_{\pm 3.82}$  & $13.65_{\pm 0.11}$ & $59.65_{\pm 30.82}$  \\
6 & $10.10_{\pm 0.12}$          & $77.60_{\pm 5.44}$  & $10.06_{\pm 0.15}$ & $66.85_{\pm 20.40}$ & $\mathbf{10.21}_{\pm 0.05}$ & $118.69_{\pm 16.79}$ \\
7 & $\mathbf{6.97}_{\pm 0.02}$  & $172.71_{\pm 66.49}$ & $6.93_{\pm 0.19}$ & $113.59_{\pm 28.29}$ & $6.85_{\pm 0.07}$ & $286.39_{\pm 66.62}$ \\
8 & $3.76_{\pm 0.68}$           & $295.45_{\pm 79.41}$ & $2.49_{\pm 0.31}$ & $200.71_{\pm 22.70}$ & $\mathbf{3.97}_{\pm 0.36}$ & $397.98_{\pm 155.61}$ \\
9 & $1.41_{\pm 0.40}$           & $243.90_{\pm 54.16}$ & $1.11_{\pm 0.88}$ & $171.49_{\pm 54.16}$ & $\mathbf{1.51}_{\pm 0.40}$ & $739.33_{\pm 486.84}$ \\
\bottomrule
\end{tabular}
}
\caption{Ablation of the expansion schedule of the Randomized HyperSteiner, for a mixture of $\mathcal{G}(\mu_{d, k}(0.9), 0.01)$, with $k \in \{1, \cdots, d\}$ for $d \in \{4, \cdots, 10\}$, sampling 20 points per Gaussian. }
\label{tab:ab_schedule}
\end{table*}

\section{Additional Related Work on Hyperbolic Machine Learning}
\label{app:hyp_ml}

Hyperbolic geometry has emerged as a powerful tool in machine learning, particularly for modeling data with latent hierarchical structure. Early work focused on embedding techniques that leverage manifold learning principles \citep{nickel2017poincare, chamberlain2017neural, Sala2018_RepresentationTradeoffsHyperbolicEmbeddings, Chami_2021_HoroPCA, Guo_2022_CO-SNE}. Following the introduction of hyperbolic neural networks \citep{ganea2018hyperbolic}, subsequent work explored hyperbolic generative modeling \citep{mathieu2019continuous, Nagano2019wrapped, bose2020latent, huang2022riemannian}. Recent developments have explored large-scale architectures, including hyperbolic large language models \citep{he2025helmhyperboliclargelanguage, He2025HyperbolicFoundationModels} enabled by adapting the Transformer architecture \citep{Vaswani2017Transformer} to hyperbolic geometry \citep{Gulcehre2018HyperbolicAttentionNetworks, Chen2022FullyHyperbolicNN, Shimizu2021HyperbolicNNPlusPlus, Yang2024Hypformer}.

These developments are motivated by hyperbolic geometry's representational advantages for hierarchical data, where exponential growth patterns can be embedded with arbitrarily low distortion compared to Euclidean space \citep{sarkar2011low}. However, while existing hyperbolic machine learning models effectively capture hierarchical structure in latent form, they do not inherently provide mechanisms for explicitly extracting the underlying hierarchy. Our work addresses this gap by building upon the \citet{garcia2025hypersteiner} framework to enable direct recovery of minimal trees from hyperbolic embeddings, supporting what may be broadly termed \textit{hyperbolic geometric inference}. Related efforts, such as \citet{Medbouhi2024hHyperDGA}, also fall under this paradigm by enabling statistical data analysis in hyperbolic latent spaces.

\end{document}